\documentclass[preprint,12pt]{elsarticle}

\usepackage{slashbox} 
\usepackage{subfig}
\usepackage{amssymb,amsmath}%
\usepackage{gensymb}
\usepackage{amsthm,bm}%
\usepackage[width=17cm,height=23.5cm,centering]{geometry}%
\usepackage{color}
\usepackage{stmaryrd} 
\usepackage{cases,empheq,algorithm,enumitem}
\usepackage[colorlinks,citecolor= blue,linkcolor= blue]{hyperref}
\usepackage{todonotes}  

\newtheorem{Theorem}{Theorem}

\newtheorem{Property}{Property}
\newtheorem{Corollary}{Corollary}

\newcommand{\ds}{\displaystyle}

\newcommand{\jump}[1]{\left\llbracket #1 \right\rrbracket}
\newcommand{\mean}[1]{\left\langle #1 \right\rangle}

\definecolor{rouge}{rgb}{1,0,0}
\definecolor{bleu}{rgb}{0,0,1}

\definecolor{colorcb}{rgb}{0.6350, 0.0780, 0.1840}

\definecolor{colorcbq}{rgb}{0, 0.5, 0}

\graphicspath{{Figures/}}

\begin{document}
	
	\begin{frontmatter}
		
		\title{Surface impedance and topologically protected interface modes \\in one-dimensional phononic crystals}
		
		\author[LMA]{Antonin Coutant}
		\ead{coutant@lma.cnrs-mrs.fr}
		\author[LMA]{Bruno Lombard\corref{cor1}}
		\ead{lombard@lma.cnrs-mrs.fr}
		\cortext[cor1]{Corresponding author. Tel.: +33 491 84 52 42 53.}
		\address[LMA]{Aix Marseille Univ, CNRS, Centrale Marseille, LMA UMR 7031, Marseille, France}
		
		\begin{abstract}
			When semi-infinite phononic crystals (PCs) are in contact, localized modes may exist at their boundary. The central question is generally to predict their existence and to determine their stability. With the rapid expansion of the field of topological insulators, powerful tools have been developed to address these questions. In particular, when applied to one-dimensional systems with mirror symmetry, the bulk-boundary correspondence claims that the existence of interface modes is given by a topological invariant computed from the bulk properties of the phononic crystal, which ensures strong stability properties. This one-dimensional bulk-boundary correspondence has been proven in various works. Recent attempts have exploited the notion of surface impedance, relying on analytical calculations of the transfer matrix. In the present work, the monotonic evolution of surface impedance with frequency is proven for all one-dimensional phononic crystals with mirror symmetry. This result allows us to establish a stronger version of the bulk-boundary correspondence that guarantees not only the existence but also the uniqueness of a topologically protected interface state. The method is numerically illustrated in the physically relevant case of PCs with imperfect interfaces, where analytical calculations would be out of reach. 
		\end{abstract}
		
		
	\end{frontmatter}
	
	
	
	\tableofcontents
	
	\section{Introduction}\label{SecIntro}
	
	Periodic media have forbidden frequency bands, in which waves cannot propagate. This property is very much used in various fields of wave physics, from quantum mechanics~\cite{Ashcroft} to electromagnetism~\cite{Joannopoulos} or acoustics~\cite{Laude}. It has been the subject of a great deal of theoretical work, particularly in spectral theory~\cite{ReedIV}. When two phononic/photonic crystals (PCs) are joined, interface modes can exist in the gaps common to both crystals, depending on the coupling between the media~\cite{Allaire02}. These modes then remain localized (in one dimension (1D)) or can propagate at the interface (in 2D or more). However, their stability is not guaranteed. In the presence of manufacturing defects or impurities, these waves are then diffracted in all directions, losing the benefit of guidance.
	
	Recently, the notion of {\it topological insulators}, initially exhumed in the context of the quantum Hall effect, has provided a powerful approach to obtain localized modes with high robustness against defects~\cite{Ozawa19,Ma19}. If certain symmetries are satisfied, appropriate topological invariants can be obtained from the bulk properties of the material~\cite{Asboth16,ProdanShultzBaldes}. The \emph{bulk-edge correspondence} then allows one to relate these invariants with the presence of edge or interface localized modes~\cite{Delplace11,Delplace21}. Despite the broad applicability of the principle, a precise statement of the bulk-boundary correspondence depends on the symmetry class, the dimensionality and the type of topological invariant. For this reason, providing a proof of this correspondence requires case-by-case analyses, with the most extensively studied case being 2D systems in the unitary class characterized by a Chern number~\cite{Hatsugai93,Essin11,ProdanShultzBaldes}. 
	
	In 1D continuous systems with mirror symmetry, the most popular invariant is the Zak phase~\cite{Zak89}, which can only take values $0$ or $\pi$ and remains constant until a Dirac point is reached. In this context, the Zak phases of contacting crystals govern the existence of localized modes. These modes are topologically protected in the sense that they are maintained by symmetry preserving continuous deformations of the medium, unless the gap is closed (Dirac point). In this context, two classes of methods have been used to analyze the existence of topological states and prove the bulk-edge correspondence. A first method consists of perturbing the Hamiltonian of the system studied, in the vicinity of a Dirac point \cite{Fefferman12,Lin22}. It therefore requires an \emph{a priori} knowledge of Dirac points or their artificial construction by band-folding. A second approach consists in focusing on specific models, where the correspondence can be established by direct calculations~\cite{Xiao14}. Moreover, several works have exploited the concept of impedance~\cite{Xiao14,Lin22}, which allows one to relate the existence of interface modes to bulk properties. This approach is elegant and offers a new point of view. However, it has been so far limited to specific examples or perturbative approaches.  The objective of the present paper is to generalize the impedance approach to any PCs. This leads us to a proof the a stronger version of the bulk-edge correspondence for 1D continuous systems with mirror symmetry. We also point out that several authors have analyzed the somewhat related but different problem of localized modes for 1D PCs with parameter dependent interfaces, and their relations to Chern numbers~\cite{Gontier20,Drouot21}. 
	
	For this purpose, the sketch of the study is as follows. Section \ref{SecPS} describes elastic PCs, where classical band structure results are recalled. Section \ref{SecCS} focuses on the particular case of mirror symmetric PCs. The symmetry properties of Bloch modes at band edges are proven, based on some existing results. Section \ref{SecZ} investigates two semi-infinite PCs in contact. The concept of surface impedance is introduced. When the frequency varies in a gap, a strict decay property of the surface impedance is proven (Property \ref{PropdZdOm}). This key result allows to prove Theorem \ref{TheoZ}, which recovers the existing results on topologically protected interface modes. 
	The absence of such modes with Neumann or Dirichlet conditions is also explained, contrary to the discrete case such as SSH. Generalizations are also proposed, notably allowing to treat the case of subwavelength PCs with Helmholtz resonators \cite{Zhao21}. Section \ref{SecNum} illustrates numerically the theoretical findings. Lastly, Section \ref{SecConclu} concludes the paper and proposes some future directions of investigations.

	
	\section{Problem statement}\label{SecPS}
	
	This section starts with a presentation of the physical problem addressed. After that and for completeness, we collect several known results about the spectrum of a periodic differential operator. The readers are referred to \cite{ReedIV,Kuchment93,Brown12} for more details about the Bloch-Floquet theory.
	
	\subsection{Physical modeling}\label{SecPS-phys}
	
	We consider linear elastic wave propagation at a given angular frequency $\omega=2\pi f$. The medium is $h$-periodic with mass density $\rho(x)$ and Young's modulus $E(x)$. The displacement $u(x)$ satisfies the Helmholtz equation
	\begin{equation}
		\frac{d}{dx}\left(E(x)\,\frac{du}{dx}\right)+\rho(x)\,\omega^2 \,u=0.
		\label{Helmholtz}
	\end{equation}
	The equation \eqref{Helmholtz} is generic for different wave physics. For instance, the case of acoustics is obtained changing $u$ by the acoustic pressure $p$, $E$ by $1/\rho$, and $\rho$ by $1/\kappa$, where $\kappa\equiv \rho c^2$ is the modulus of compressibility. Similarly, the case of photonics is obtained by changing $\rho$ by the permittivity $\epsilon$ and $E$ by $1/\mu$, where $\mu$ is the permeability. Without loss of generality, the edges of the periodic cells are assumed to be located at $nh$, with $n\in \mathbb{Z}$. The physical parameters are piecewise smooth, $L^\infty$ and strictly positive. 
	
	
	\subsection{Band structure}\label{SecPS-spectrum}
	
	We denote $L^2(\mathbb{R})$ the Hilbert space equipped with the inner product
	\begin{equation}
		\mean{f|g}=\int_0^h\rho(x)\,f(x)\,\overline{g}(x)\,dx , 
		\label{ProdScal}
	\end{equation}
	where $\overline{g}$ refers to the complex conjugate of $g$. Let ${\cal B}=[-\pi,\pi]$ be the first Brillouin zone. The Bloch Hamiltonian ${\cal L}$ is defined as the linear operator 
	\begin{equation}
		{\cal L}u=-\frac{1}{\rho(x)}\,\frac{d}{dx}\left(E(x)\,\frac{du}{dx}\right) \mbox{  for } x\in\mathbb{R}. 
	\end{equation} 
	Given $q\in{\cal B}$, Bloch-Floquet conditions yield a $q$-dependent eigenvalue problem on  $\mathbb{R}$: find $(\varepsilon,u)$ such that 
	\begin{equation} \label{Floquet_Eig}
		\left\{	
		\begin{aligned}
			{\cal L}(q) u &= \varepsilon u , \\
			u(x+h) &= e^{\mathrm{i}q}\,u(x) ,
		\end{aligned}
		\right. 
	\end{equation} 
	in the function space
	$$
	L_q^2=\{g \in L^2_{\rm loc}:g(x+h)=e^{\mathrm{i}q}\,g(x)\}.
	$$
	The eigenvalue $\varepsilon$ is the square angular frequency that appears in the Helmholtz equation~\eqref{Helmholtz}: $\varepsilon = \omega^2$. The eigenvalue problem~\eqref{Floquet_Eig} has a an infinite set of (possibly repeated) real positive eigenvalues $\varepsilon_n(q):=\omega_n^2(q)$ ordered by increasing value: 
	$$
	\varepsilon_1(q)\leq \varepsilon_2(q)\leq \cdots \leq \varepsilon_n(q)\leq \cdots
	$$
	The dispersion relations of the $n$th band $q \to \varepsilon_n(q)$ are continuous with respect to $q\in{\cal B}$ and monotonous on each half of the Brillouin zone, i.e. $[-\pi, 0]$ and $[0,\pi]$~\cite{ReedIV,Shahraki22}. Moreover, they satisfy $\varepsilon_n(-q)=\varepsilon_n(q)$, as imposed by reciprocity. The eigenfunctions $u_n$ are called Bloch modes, and they are orthogonal for the inner product \eqref{ProdScal}. Lastly, $\varepsilon_1(0)=0$, and the corresponding Bloch mode is a constant function. For each integer $n$, let
	$$
	\varepsilon_n^-=\min\{\varepsilon_n(q):q\in{\cal B}\},\qquad \varepsilon_n^+=\max\{\varepsilon_n(q):q\in{\cal B}\}.
	$$
	Then the entire spectrum of the Hamiltonian is given by
	$$
	\sigma({\cal L})=\bigcup_{n\geq 1}[\varepsilon_n^-,\varepsilon_n^+],
	$$
	which corresponds to the essential part of the spectrum. The interval $[\varepsilon_n^+,\varepsilon_{n+1}^-]$ is a gap if $\varepsilon_n^+<\varepsilon_{n+1}^-$ for some $n$. 
	
	
	\section{Phononic crystals with mirror symmetry}\label{SecCS}
	
	\subsection{Parity of Bloch modes}\label{SecCS-parity}
	
	From now on, we consider a particular case of periodic media, which are mirror symmetric, or equivalently reflection symmetric with respect to the centre of each cell. On $[0,h]$, we thus have $\rho(x)=\rho(h-x)$ and $E(x) = E(h-x)$. One introduces the parity operator ${\cal P }$, such that $({\cal P}f)(x)=f(h-x)$. Mirror symmetry implies that 
	\begin{equation}
		{\cal P}\,{\cal L}(-q)={\cal L}(+q)\,{\cal P}, \qquad \forall q\in{\cal B}.
		\label{CommutPL}
	\end{equation}
	Since ${\cal L}(-\pi)={\cal L}(\pi)$, it follows from \eqref{CommutPL} that the Hamiltonian and the parity operator commute at the band edges $q=0$ and $q=\pm \pi$.
	
	Let us consider an eigenvector $u_n(q)$ of ${\cal L}(q)$, with eigenvalue $\varepsilon_n(q)$. Then \eqref{CommutPL} gives that ${\cal P} u_n(q)$ is an eigenvector of ${\cal L}(-q)$ with the same eigenvalue $\varepsilon_n(q)$. Moreover, $u_n(-q)$ is also an eigenvector of ${\cal L}(-q)$ with eigenvalue $\varepsilon_n(q)$. Assuming that $\varepsilon_n(q)$ is a non-degenerate eigenvalue (non-overlapping bands), it means that $u_n(-q)$ and ${\cal P} u_n(q)$ are proportional to one another. Since ${\cal P}$ is unitary, it follows
	\begin{equation}
		u_n(-q)=\mu \,{\cal P} u_n(q)=e^{\mathrm{i}\,\xi_n(q)}{\cal P}u_n(q),
		\label{ThetaQ}
	\end{equation}
	where $\xi_n(q)$ is a locally smooth function of $q$. At the band edges ($q=0,\pi$), applying ${\cal P}$ on both sides of \eqref{ThetaQ} and using the symmetry property ${\cal P}^2=1$ gives $\mu^2=1$. As a consequence $\mu=\{\pm 1\}$ and $\xi_n=\{0,\pi\}$ at the band edges. The case $\mu=+1$ (and $\xi_n=0$) gives $u_n(q)={\cal P} u_n(q)$: the Bloch mode is symmetric, and $u'_n$ is anti-symmetric, where prime denotes the space derivative. Conversely, in the case $\mu=-1$ (and $\xi_n=\pi$), then $u_n(q)=-{\cal P} u_n(q)$: the Bloch mode is anti-symmetric, and $u'_n$ is symmetric. We introduce the spaces of symmetric functions ${\cal S}$ and anti-symmetric functions ${\cal A}$
	\begin{equation}
		{\cal S}=\left\{g : {\cal P}g=+g\right\},\qquad {\cal A}=\left\{g : {\cal P}g=-g\right\}.
		\label{SA}
	\end{equation}
	Then the properties of Bloch modes at the edges of the gaps (or simply "band edges") are summed up as follows.
	
	\begin{Property}
		Let us consider a mirror symmetric PC. We assume that the $n$th eigenvalues of the Hamiltonian are not degenerate at $q=0$ and $\pi$. Then the following alternative holds at the band edges: 
		\begin{itemize}
			\item either $u_n \in {\cal S}$ and $u'_n\in {\cal A}$;
			\item or $u_n \in {\cal A}$ and $u'_n\in {\cal S}$.
		\end{itemize} 
		\label{PropCS}
	\end{Property} 
	
	\noindent
	The values of $u_n(0)$ and $u'_n(0)$ are deduced from the Property \ref{PropCS}. First of all, $u_n(0)$ and $u'_n(0)$ cannot be simultaneously zero (and similarly at $h$). By uniqueness of the solution of \eqref{Helmholtz}, this would imply that $u_n$ is identically null. 
	
	\begin{table}[h!]
		\begin{center}
			\begin{small}
				{\renewcommand{\arraystretch}{1.5}
					\begin{tabular}{c||c|c}
						& $q=0$ & $q=\pi$ \\ [4pt]
						\hline 
						$u_n\in{\cal S}$ & $u'_n(0)=0$ & $u_n(0)=0$ \\ [4pt]
						\hline
						$u_n\in{\cal A}$ & $u_n(0)=0$ & $u'_n(0)=0$ \\ [4pt]
						\hline
					\end{tabular}
				}
			\end{small} 
		\end{center}
		\vspace{-0.5cm}
		\caption{\label{TabU}Values of the Bloch modes $u_n(0)$ and $u'_n(0)$ at the edges of the $n$th gap. The cases are distinguished according to the symmetry of $u_n$ and the value of $q$.}
	\end{table}
	Then, let us first consider $q=0$ and assume that $u_n\in {\cal S}$. The Property \ref{PropCS} implies $u'_n\in{\cal A}$, and hence $u'_n(h)=-u'_n(0)$. But the Bloch-Floquet condition gives $u'_n(h)=e^{\mathrm{i}0}u'_n(0)=+u'_n(0)$, and thus $u'_n(0)=-u'_n(0)\equiv 0$. If $u_n\in{\cal A}$, then similar arguments yield $u_n(0)=-u_n(0) = 0$. At $q=\pi$, the Bloch-Floquet conditions gives $u_n(h)=e^{\mathrm{i}\pi}u_n(0)=-u_n(0)$ and $u'_n(h)=-u'_n(0)$. If $u_n\in{\cal S}$ then $u_n(0)=0$. Conversely, if $u_n\in{\cal A}$ then $u'_n\in{\cal S}$ and $u'(0)=0$. All these cases are summarized in the Table \ref{TabU}. 
	
	Figure \ref{FigBloch} illustrates these different cases. The physical and geometrical parameters correspond to the configuration studied in Section \ref{SecNum} (with the interface spacing $\theta=0.25$), and will be described later. For the moment, it is sufficient to observe the symmetry properties. Figure \ref{FigBloch} represent the Bloch mode at $q=\pi$ and at frequencies of the band edges of the first gap. At the entry frequency $\varepsilon_1^+$ of gap 1 (a), $u_1$ is symmetric. As written in Table \ref{TabU}, $u_1$ is therefore symmetric and cancels at $x=0$ and $x=h$. On the contrary, $u_2$ is antisymmetric at the exit frequency $\varepsilon_2^-$ of the gap 1 (b), leading to $u_2'(0)=u_2'(h)=0$.
	
	\begin{figure}[h!]
		\begin{center}
			\begin{tabular}{cc}
				(a) & (b)\\
				\hspace{-1cm}
				\includegraphics[width=0.55\textwidth]{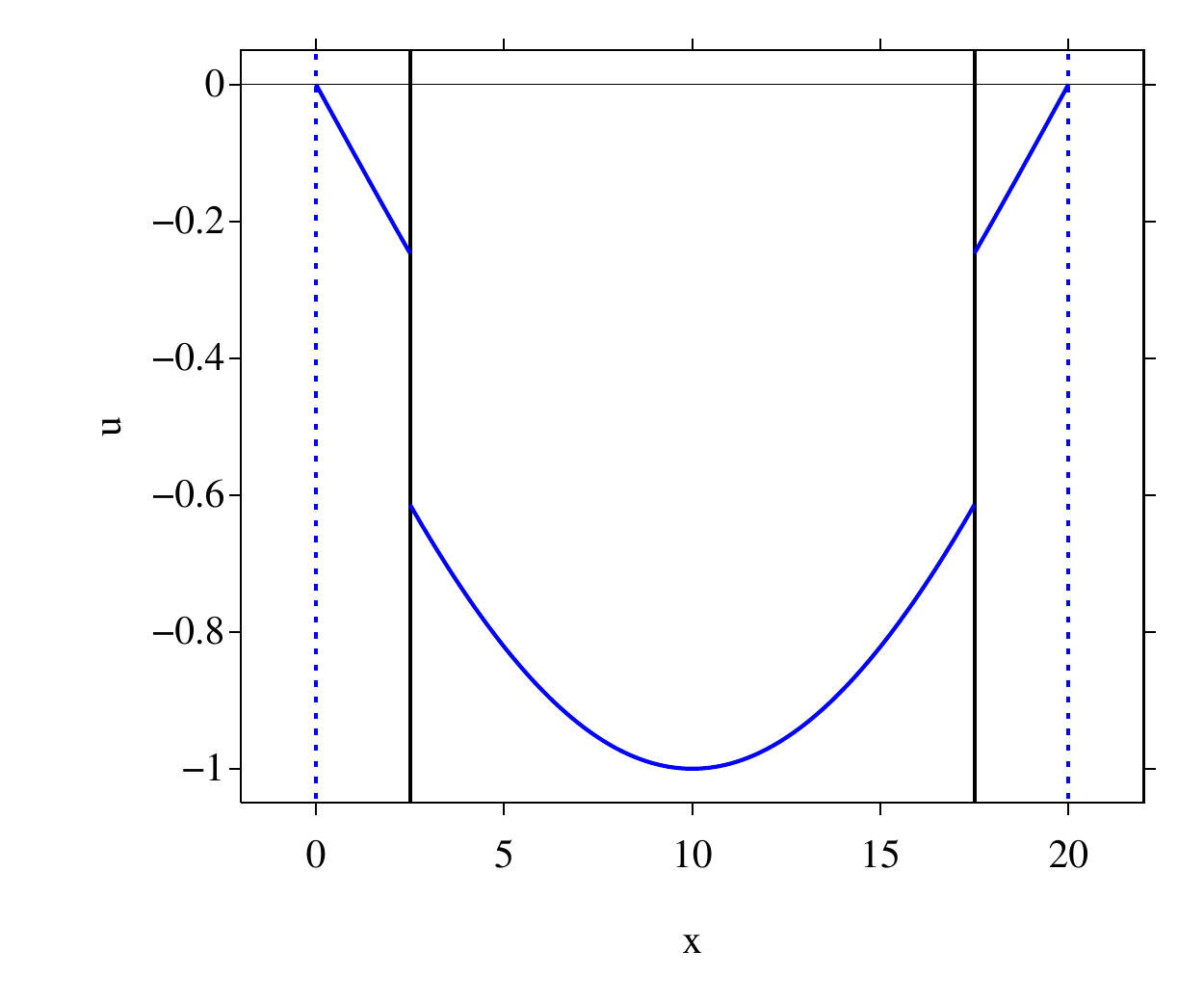} &
				\hspace{-1cm}
				\includegraphics[width=0.55\textwidth]{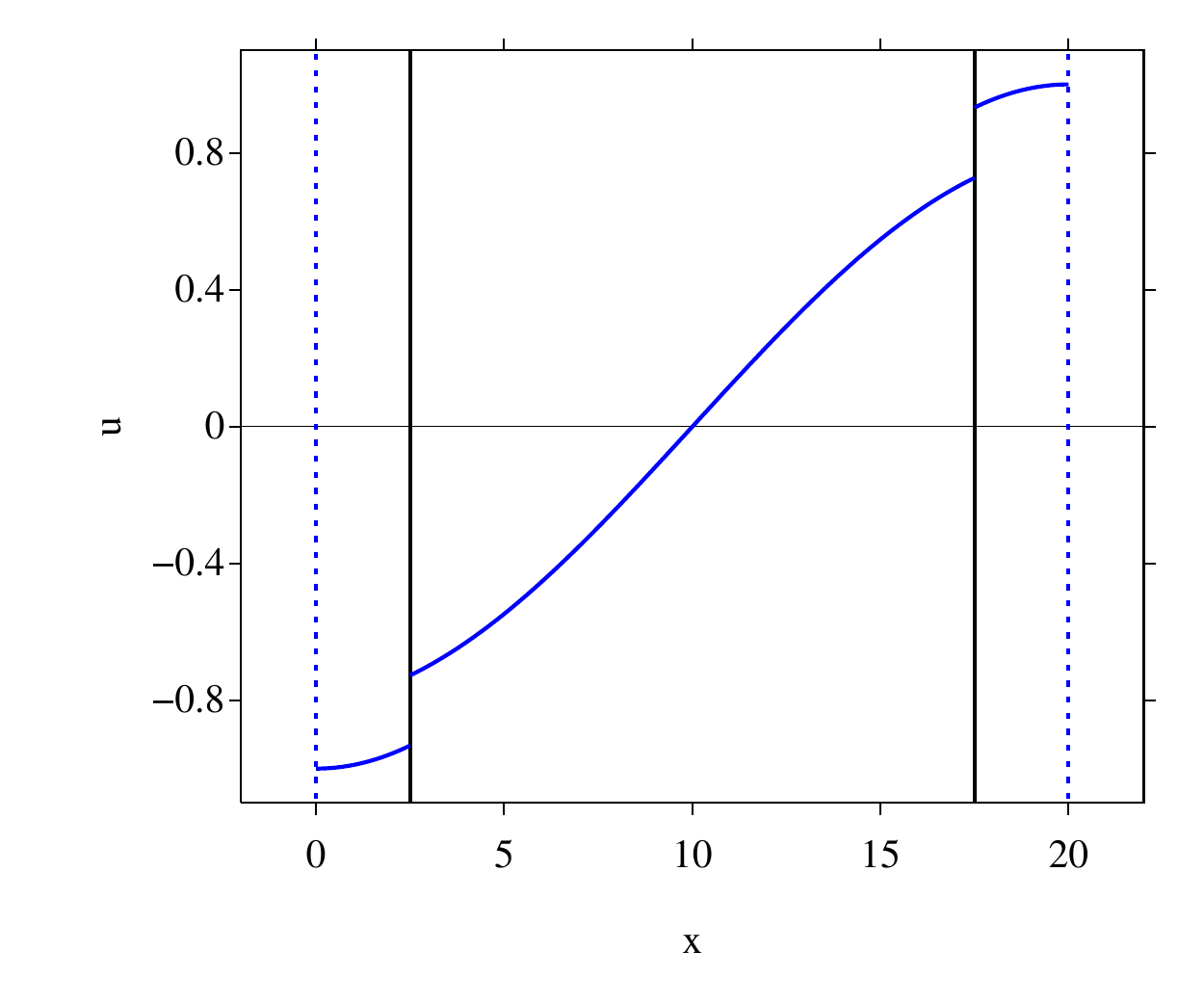} 
			\end{tabular}
		\end{center}
		\vspace{-0.6cm}
		\caption{\label{FigBloch}Bloch mode at the edges of gap 1 ($q=\pi$). (a): entry of the gap $\varepsilon_1^+$; (b): exit of the gap $\varepsilon_{2}^-$. The vertical solid lines denote the imperfect interfaces. The vertical dotted lines denote the edges of the elementary cell. The physical parameters are described in Section \ref{SecNum}.} 
	\end{figure}
	
	
	\subsection{Change of parity across a gap}\label{SecCS-change}
	
	Here we investigate the change of parity of the Bloch modes at $q=0$ or $q=\pi$ when the frequency crosses a gap. The proof relies on the oscillation theory of Sturm-Liouville operators, see for instance \cite{Weidmann06,Brown12,Lin22}. Let us denote $\varepsilon_j^P$, $\varepsilon_j^A$, $\varepsilon_j^D$, and $\varepsilon_j^N$ the $j$th eigenvalues of the Hamiltonian ${\cal H}$ on the elementary cell $[0,h]$, respectively associated with the following boundary conditions:
	\begin{enumerate}
		\item periodic boundary conditions: $u(h)=u(0)$, $u'(h)=u'(0)$, which is the Bloch condition at $q=0$;
		\item antiperiodic boundary conditions: $u(h)=-u(0)$, $u'(h)=-u'(0)$, which is the Bloch condition at $q=\pi$;
		\item Dirichlet boundary conditions: $u(h)=u(0)=0$;
		\item Neumann boundary conditions: $u'(h)=u'(0)=0$.
	\end{enumerate}
	As stated in Theorem 13-10 of \cite{Weidmann06}, these eigenvalues satisfy the interlacing property:
	\begin{equation}
		\begin{array}{l}
			\ds \varepsilon_1^N\leq \varepsilon_1^P < \varepsilon_1^A \leq \{\varepsilon_2^N,\varepsilon_1^D\} \leq \varepsilon_2^A < \varepsilon_2^P \leq \{\varepsilon_3^N,\varepsilon_2^D\} \leq \cdots \\ [6pt]
			\ds \cdots \leq \varepsilon_{2n-1}^P < \varepsilon_{2n-1}^A \leq \{\varepsilon_{2n}^N,\varepsilon_{2n-1}^D\} \leq \varepsilon_{2n}^A < \varepsilon_{2n}^P \leq \{\varepsilon_{2n+1}^N,\varepsilon_{2n}^D\} < \varepsilon_{2n+1}^P \leq \cdots
		\end{array}
		\label{Interlacing}
	\end{equation}
	In substance, this says that there is a single Neumann and a single Dirichlet eigenvalue inside each gap (possibly closed) of the corresponding periodic problem. This leads to the following result (see Theorem 4.4 of \cite{Lin22}): 
	
	\begin{Property}
		Let us consider a mirror symmetric PC, where the $j$th band is isolated. Then the Bloch modes on each edge of a gap, i.e. $(q,\varepsilon_j^+)$ and $(q,\varepsilon_{j+1}^-)$ with $q=0$ or $q=\pi$, attain different symmetries.
		\label{PropParityChange}
	\end{Property}
	
	\noindent
	In other words, if $u_n\in{\cal S}$ then $u_{n+1}\in{\cal A}$, and inversely. To show this, we first notice that the boundary values in Table \ref{TabU} imply that the Bloch mode on each edge of the $n$th gap satisfy either Neumann or Dirichlet boundary conditions. Using the interlacing property \eqref{Interlacing}, we conclude that there is a single Neumann and a single Dirichlet eigenvalue on the gap edges. It follows Property~\ref{PropParityChange}. 
	
	
	\subsection{Relation between Zak phase and Bloch mode symmetry}\label{SecCS-Zak}
	
	Bands of a mirror symmetric 1D PC possess topological properties characterized by an invariant called the Zak phase. To define it, we introduce the Berry connection $A_n(q)$ of the $n$th band: 
	\begin{equation}
		A_n(q)=-\mathrm{i} \mean{u_n(q)|\partial_q u_n(q)},
		\label{Berry}
	\end{equation}
	where $\mean{.|.}$ is the Hermitian inner product~\eqref{ProdScal}. The integral of the Berry connection across the first Brillouin zone gives the Zak phase:
	\begin{equation}
		\Phi_n=\int_{-\pi}^{+\pi}A_n(q)\,dq.
		\label{Zak}
	\end{equation}
	Notice that the Zak phase is defined using the Bloch mode $u_n(q)$, which is not periodic in $x$, but periodic in $q$ (see equation~\eqref{Floquet_Eig}). In the case of a periodic and mirror symmetric medium, and assuming that the $n$th eigenvalue is non-degenerate (no Dirac point), the Zak phase is directly given by the change of symmetry from one edge of the band to the other, namely: 
	\begin{equation}
		\Phi_n=\xi_n(\pi)-\xi_n(0),
		\label{PhiZak}
	\end{equation} 
	where $\xi$ is defined in \eqref{ThetaQ}. A proof of \eqref{PhiZak} is given in \ref{Sec-ZakTopo}. The Zak phase is defined mod $2\pi$, and it can be either 0 (trivial) or $\pi$ (topological).
	
	As we shall see in Section \ref{SecZ}, the existence of an interface localized mode depends on the symmetries on each edge of a given gap, i.e. at $\varepsilon_n^+$ and $\varepsilon_{n+1}^-$. As we already saw (see Property~2), the two symmetries are necessarily different, so it is sufficient to know the mode symmetry at $\varepsilon_n^+$. Hence, we define the bulk topological index of the gap $n$ as ${\cal J}_n$:
	\begin{equation}
		{\cal J}_n=\left\{
		\begin{array}{ll}
			+1 & \mbox{ if } u_n^+\in {\cal S},\\ [6pt]
			-1 & \mbox{ if } u_n^+\in {\cal A}.
		\end{array}
		\right.
		\label{IndiceJ}
	\end{equation}
	Now, since the Bloch mode at zero frequency ($\varepsilon_1(0) = 0$) is always symmetric ($u_1^- \in {\cal S}$), we can obtain the symmetry at the entry of a given gap by checking how many symmetry changes occur when the frequency ranges from 0 to $\varepsilon_n^+$: for each gap there is a change, and for each band there is a change if it is topological and none if it is trivial. As shown in \cite{Xiao14,Lin22}, this translates into 
	\begin{equation}
		{\cal J}_n=\ds (-1)^{n-1}\prod_{j=1}^n e^{\mathrm{i}\,\Phi_j}.
	\end{equation}
	
	
	\section{Semi-infinite phononic crystals}\label{SecZ}
	
	\subsection{Surface impedance}\label{SecZ-def}
	
	From now on, we consider two semi-infinite PCs in contact at $x=0$. The interface at $x=0$ is at the boundary of the unit cells for both the left and right PC, and $u$ and $E\,u'$ are continuous across the interface. The features of the PCs on the left ($x<0$) and right ($x>0$) are denoted by the indices $L$ and $R$, respectively. For instance, the Hamiltonians are ${\cal L}_L$ and ${\cal L}_R$, with eigenvalues $\varepsilon_{n,L}$ and $\varepsilon_{n,R}$. In each PC, the parameters are spatially periodic and with mirror symmetry. For simplicity, we assume that the periods are identical and equal to $h$, without this being restrictive.  
	
	Let us assume that the intersection of the $n$th gaps of the two Hamiltonians is non empty: 
	\begin{equation}
		\Omega_n=[\varepsilon_{n,L}^+,\varepsilon_{n+1,L}^-]\cap[\varepsilon_{n,R}^+,\varepsilon_{n+1,R}^-]:=[\varepsilon_n^+,\varepsilon_{n+1}^-]\neq \emptyset.
		\label{OmegaN}
	\end{equation} 
	For simplicity, we assumed the left and right gaps to be labeled by the same number $n$, but our results stand for any two gaps $n_L$ and $n_R$ as long as the intersection $\Omega$ as in equation~\eqref{OmegaN} is non-empty. For further use, $u_{n,L}^+$ denotes the Bloch mode for the eigenvalue $\varepsilon_{n,L}^+$; similarly $u_{n,R}^+$ denotes the Bloch mode for the eigenvalue $\varepsilon_{n,R}^+$.
	
	The wave fields with frequency in this gap, solution of equation~\eqref{Helmholtz}, can be decomposed into two modes, one increasing exponentially and the other decreasing exponentially as $|x|\rightarrow+\infty$. For frequencies inside $\Omega_n$, one defines {\it surface impedances} of PC-L and PC-R \cite{Xiao14}:
	\begin{equation}
		Z_L(\omega)=-\frac{u(0^-,\omega)}{E\,u'(0^-,\omega)} , \qquad 
		Z_R(\omega)=+\frac{u(0^+,\omega)}{E\,u'(0^+,\omega)} , 
		\label{ImpedanceZ}
	\end{equation}
	where $u$ is the unique solution decreasing when $x \to -\infty$ for $Z_L$ and when $x\to +\infty$ for $Z_R$. Inside a gap, $u$ and $u'$ are real, and hence $Z_{L,R}$ are real. Hereafter, we will simply write $Z$ to characterise the generic properties of the surface impedance, independently of the medium considered. We will write the frequency indifferently $\varepsilon$ or $\omega$, with $\varepsilon=\omega^2$. Lastly, the notation $Z_n$ will denote the surface impedance at $\varepsilon_n$.  
	
	The notion of surface impedance will be fundamental in the rest of the paper. An interface mode at $x=0$ corresponds to two evanescent modes satisfying the continuity conditions $u(0^-,\omega)=u(0^+,\omega)$ and $E\,u'(0^-,\omega)=E\,u'(0^+,\omega)$. This leads to the following property.
	
	\begin{Property}
		A necessary and sufficient condition for the existence of an interface localized mode in $\Omega_n$ is
		\begin{equation}
			Z_L(\omega)+Z_R(\omega)=0, \qquad \omega\in\Omega_n.
			\label{ZLZR}
		\end{equation}
		\label{PropZLZR}
	\end{Property}
	
	\noindent
	To determine the existence of a solution to \eqref{ZLZR}, we study the properties of $Z$ in the gaps. The first key element is that the impedance value at the edges of a gap is given by the symmetries of Bloch modes. Indeed, at the edges of a gap, $u:=u_n$ and $u':= u'_n$. Based on the definition of $Z$ and on Table \ref{TabU}, the surface impedance $Z_n$ on the edges of a gap is either $0$ or $\infty$, depending on the symmetry of $u_n$. This is summarized in Table~\ref{TabZ}.  
	
	\begin{table}[h!]
		\begin{center}
			\begin{small}
				\begin{tabular}{c||c|c}
					& $q=0$ & $q=\pi$ \\ [6pt]
					\hline
					$u_n\in{\cal S}$ & $Z_n=\pm\infty$ & $Z_n=0$\\ [6pt]
					$u_n\in{\cal A}$ & $Z_n=0$ & $Z_n=\pm\infty$\\
					\hline
				\end{tabular}
			\end{small} 
		\end{center}
		\vspace{-0.5cm}
		\caption{\label{TabZ}Surface impedance $Z_n$ on the edges of the Brillouin zone and at the edges of the $n$th gap, depending on the symmetry of the Bloch mode $u_n$.}
	\end{table} 
	
	At this level, it is tempting to use a continuity argument to find solutions of equation~\eqref{ZLZR}, i.e. interface modes. Table \ref{TabZ} is however not enough to conclude, since we don't control the sign of the impedance. To do so, we will now investigate the evolution with the frequency of $Z_L$ and $Z_R$ inside a gap.

	
	\subsection{Monotony of the surface impedances in gaps}\label{SecZ-BG}
	
	We now discuss a key new result, which will allow us to make strong statement regarding the bulk-boundary correspondence: the in-gap impedance obtained from equation~\eqref{Helmholtz} is always a decreasing function of the frequency. For this, purpose, one differentiates the Helmholtz equation \eqref{Helmholtz} with respect to $\omega$. Setting $\varphi=\frac{du}{d\omega}$, one obtains
	\begin{equation}
		\frac{d}{dx}\left(E(x)\frac{d\varphi}{dx}\right)+\rho(x)\,\omega^2\,\varphi=-2\,\rho(x)\,\omega\,u.
		\label{HelmholtzDer}
	\end{equation}
	From the definition of $Z_R$ in \eqref{ImpedanceZ}, it follows
	\begin{equation}
		\frac{dZ_R}{d\omega}=\frac{W(0^+,\omega)}{\left(E\,u'(0^+)\right)^2},
		\label{dZdOm}
	\end{equation}
	with the Wronskian
	\begin{equation}
		{\cal W}(x,\omega)=E(x) \left(\varphi(x) \,u'(x ) -u(x)\,\varphi'(x)\right).
		\label{Wronskian}
	\end{equation}
	The prime means the derivative with respect to $x$. The Wronskian is then differentiated with respect to $x$. Using \eqref{Helmholtz} and \eqref{HelmholtzDer} gives
	\begin{equation}
		\begin{array}{lll}
			\ds {\cal W}'& = & \ds \varphi\,\left(E\,u'\right)'-u\,\left(E\,\varphi'\right)',\\ [10pt]
			&=& \ds \varphi\,\left(-\rho\,\omega^2\,u\right)-u\,\left(-2\,\rho\,\omega\,u-\rho\,\omega^2\,\varphi\right),\\ [10pt]
			&=& \ds 2\,\rho\,\omega\,u^2.
		\end{array}
		\label{dZdX}
	\end{equation}
	The latter is integrated with respect to $x$ on PC-R:
	\begin{equation}
		{\cal W}(+\infty,\omega)-{\cal W}(0,\omega)=2\,\omega\int_0^{+\infty}\rho\,u^2\,dx.
	\end{equation}
	In the gap, the evanescent field vanishes when $x\rightarrow+\infty$, and hence $W(+\infty,\omega)=0$. Using \eqref{dZdOm} and the fact that $u$ is real in a gap leads to
	\begin{equation}
		\frac{dZ_R}{d\omega}=-\frac{2\,\omega}{\left(Eu'(0^+)\right)^2}\int_0^{+\infty}\rho\,u^2\,dx<0.
		\label{dZRdOm}
	\end{equation}
	A similar argument is used to prove that the left surface impedance decreases: $\frac{dZ_L}{d\omega}<0$ in a gap. These results are summed up in the next property.
	
	\begin{Property}
		In the gap $\Omega_n$, the surface impedance decreases with frequency:
		\begin{equation}
			\frac{dZ}{d\omega}<0,\qquad \omega\in \Omega_n.
			\label{dZdOm}
		\end{equation}
		\label{PropdZdOm}
	\end{Property}
	
	\noindent
	Property \ref{PropdZdOm} implies that $Z$ never vanishes and never becomes infinite inside a gap. Indeed, let us assume for instance that $u_n\in{\cal S}$ at $q=0$. Table \ref{TabZ} states that $Z:=Z_n=\pm \infty$ at $\varepsilon_n^+$. In the gap $\Omega_n$, $Z$ decreases monotonically, and Property \ref{PropParityChange} states that $u_{n+1}\in {\cal A}$, so that $Z=0$ at $\varepsilon_{n+1}^-$. It fixes the sign $Z_n=+ \infty$ at $\varepsilon_n^+$ and ensures that $Z\neq 0$ and $Z \neq \infty$ on $]\varepsilon_n^+,\varepsilon_{n+1}^-[$. A similar argument can be used with $u_n\in{\cal A}$ at $q=0$, and also at $q=\pi$. An alternative proof of this property based on the transfer matrix is proposed in \ref{Sec-TM}. 
	
	
	\subsection{Topologically protected interface mode}\label{SecZ-EdgeState}
	
	Now we are ready to state the main result of this article.
	
	\begin{Theorem}
		Let two mirror symmetric PCs in perfect contact at $x=0$ with a non-empty common gap $\Omega_n$ \eqref{OmegaN}. Let also ${\cal J}_{n,L}$ and ${\cal J}_{n,R}$ be the bulk topological indices \eqref{IndiceJ} of PC-L and PC-R, respectively.
		Without change of symmetry in left and right Bloch modes, ie
		\begin{equation}
			{\cal J}_{n,L}+{\cal J}_{n,R}\neq 0,
		\end{equation}
		then no interface mode exists. In the opposite case of left and right Bloch modes with different symmetries, ie
		\begin{equation}
			{\cal J}_{n,L}+{\cal J}_{n,R}= 0,
		\end{equation}
		then there exists a unique interface mode in $\Omega_n$ at $\varepsilon_n^\sharp$. This mode is topologically protected as it is maintained by continuous transformations of the PCs as long as no Dirac point is reached.
		\label{TheoZ}
	\end{Theorem}
	
	\begin{proof}
		Let us assume that the mode symmetries are the same for each PC and two overlapping gaps  (trivial interface). For instance, $u_{n,L}^+\in{\cal S}$ and $u_{n,R}^+\in{\cal S}$ at $q=0$. Table \ref{TabZ} and Property \ref{PropdZdOm} imply that $Z_L$ decreases from $+\infty$ to 0 when the  frequency varies from $\varepsilon_{n,L}^+$ to $\varepsilon_{n+1,L}^-$. Similarly, $Z_R$ decreases from $+\infty$ to 0 when the frequency varies from $\varepsilon_{n,R}^+$ to $\varepsilon_{n+1,R}^-$. It follows that $Z_L+Z_R>0$ never vanishes on $\Omega_n$. Property \ref{PropZLZR} ensures thus that no interface mode exists. The proof in the cases $u_{n,L}^+\in{\cal A}$ and $u_{n,R}^+\in{\cal A}$, as well as $q=\pi$, follows exactly the same lines.
		
		Next one considers the case of opposite symmetries, say $u_{n,L}^+\in{\cal S}$ and $u_{n,R}^+\in{\cal A}$ at $q=0$, (topological interface). As in the previous case, $Z_L$ decreases from $+\infty$ to 0 when the frequency varies from $\varepsilon_{n,L}^+$ to $\varepsilon_{n+1,L}^-$. But the antisymmetry of $u_{n,R}^+$ implies that $Z_R$ decreases from 0 to $-\infty$ when the frequency varies from $\varepsilon_{n,R}^+$ to $\varepsilon_{n+1,R}^-$. Now, the lower edge of the gap is at $\varepsilon_{n,L}^+$ or at $\varepsilon_{n,R}^+$. At that lower edge, $Z_L+Z_R$ is always strictly positive: in the first case, $(Z_L+Z_R)_{\varepsilon = \varepsilon_{n,L}^+} = +\infty$, while in the second case $(Z_L+Z_R)_{\varepsilon = \varepsilon_{n,R}^+} = Z_L(\varepsilon_{n,R}^+)>0$. Similarly, on the upper edge, either at $\varepsilon_{n+1,L}^-$ or at $\varepsilon_{n+1,R}^-$, $Z_L+Z_R$ is always strictly negative: $(Z_L+Z_R)_{\varepsilon = \varepsilon_{n+1,L}^-} = Z_R(\varepsilon_{n+1,L}^-)<0$ in the first case and $(Z_L+Z_R)_{\varepsilon = \varepsilon_{n+1,L}^-} = -\infty$ in the second case. Hence, by continuity and monotony of $Z_L+Z_R$ (Property~\ref{PropdZdOm}), it vanishes a single time in the $\Omega_n$ interval. By Property~\ref{PropZLZR}, there is a unique interface localized mode in this frequency interval. 
		
		As indicated in Section \ref{SecCS-Zak}, the symmetries of $u_{n,L}^+$ and $u_{n,R}^+$ are maintained by continuous deformations of the $n$th band, as long as it remains isolated. Since the interface mode is a direct consequence of the symmetries of Bloch modes, it is topologically protected.
	\end{proof}
	
	
	\subsection{Absence of edge mode for one-sided systems}\label{SecZ-1sided}
	
	In the particular case of a single PC with Dirichlet boundary condition ($u(0,\omega)=0$) or Neumann boundary condition ($E\,u'(0,\omega)=0$), then only one surface impedance is involved, for instance $Z_R$. A mode localized near $x=0$ will exist if and only if $Z_R=0$ for a Neumann boundary condition and $Z_R = \infty$ for a Dirichlet one. In this case, we talk about edge modes rather than interface modes. However, at the entry of the $n$th gap $\Omega_n$, it takes the value $Z_n=0$ or $Z_n=+\infty$, respectively. $Z$ then decreases monotonically and never vanishes nor becomes infinite. This was already mentioned after Property~\ref{PropdZdOm}, but it results in the following interesting Property.
	
	\begin{Property}
		Let us consider one semi-infinite mirror symmetric PC, with Dirichlet or Neumann boundary conditions. Then no edge mode exists.
		\label{Prop1sided}
	\end{Property}
	
	\noindent
	This is a rather surprising result. First, it is in direct contrast with discrete systems (e.g. lattice models) where edge modes can be found and protected by a quantized Zak phase. Second, it shows explicitly that the bulk-boundary correspondence only works for interfaces in the case of continuous systems with mirror symmetry. A third striking consequence of Property \ref{Prop1sided} is that edge modes (with Dirichlet or Neumann boundary conditions) may exist only if one breaks the mirror symmetry (either in the bulk or by an appropriate choice of edge). This is the case for instance in the waveguide realisation of SSH \cite{Coutant21}. 
	
	
	\subsection{Generalizations}\label{SecZ-Gene}
	
	We now discuss some possible generalizations of the previous results. Our aim is to show that the bulk-boundary correspondence as established here, i.e. Theorem~\ref{TheoZ}, stands for many popular cases beyond equation~\eqref{Helmholtz}, for instance when resonators are added or with imperfect interfaces. This latter case will be used in Section~\ref{SecNum} to illustrate the findings. 
	
	On the elementary cell $[0,h]$, a set of $N$ interfaces is considered and is denoted by ${\cal I}=\{x_1,\cdots,x_N\}$. Between each interface, the Helmholtz equation takes the generalized form: 
	\begin{equation}
		\frac{d}{dx}\left(E(x)\,\frac{du}{dx}\right)+V(x,\omega)\,u =0 , 
		\label{Helmholtz-Potentiel}
	\end{equation}
	and at $x_j$, the fields satisfy the jump conditions:
	\begin{equation}
		\jump{u}_{x_j}=\alpha_j\mean{E\,\frac{du}{dx}}_{x_j},\qquad \jump{E\,\frac{du}{dx}}_{x_j}=-W_j(\omega)\mean{u}_{x_j},
		\label{JCspringmass}
	\end{equation}
	with $\alpha_j\geq 0$ and $W_j(\omega) \geq 0$. In \eqref{JCspringmass}, $\jump{.}_{x_j}$ and $\mean{.}_{x_j}$ denote the jump and the mean value at the interface $x_j$, respectively, and they are defined for any function $g(x)$ by
	\begin{equation}
		\jump{g}_{x_j}=g(x_j^+)-g(x_j^-), \qquad \mean{g}_{x_j}=\frac{1}{2}\left(g(x_j^+)+g(x_j^-)\right).
		\label{NotaSM}
	\end{equation}
	The PCs are still assumed to be mirror symmetric: $V$ is an even function of $x$ for all $\omega$, and the interfaces are placed symmetrically in each unit cells.  Between each interface, the spatial evolution of the Wronskian \eqref{dZdX} becomes
	\begin{equation}
		{\cal W}'=  \partial_\omega V \,u^2.
	\end{equation}
	Integrating across all interfaces and using the jump conditions \eqref{JCspringmass}, we obtain the frequency evolution of $Z_R$, which generalizes \eqref{dZRdOm}: 
	\begin{equation}
		\frac{dZ_R}{d\omega}=-\frac{1}{\left(Eu'(0^+)\right)^2} \left(\int_0^{+\infty}\partial_\omega V \,u^2\,dx + \sum_{n=0}^{+\infty} \sum_{j=1}^{N}\partial_\omega W_j \,u^2(x_j+nh)\right).
	\end{equation}
	The sign of $\frac{dZ_R}{d\omega}$ depends obviously on the sign of $\partial_\omega V$ and of $\partial_\omega W_j$. A similar analysis can be performed on PC-L, yielding the following generalization of Property \ref{PropdZdOm}.
	
	\begin{Corollary}
		Let us consider a mirror symmetric PC described by equation~\eqref{Helmholtz-Potentiel} and the jump conditions~\eqref{JCspringmass}. We assume that $V(x,.)$ and $W_j(.)$ are rational functions of $\omega$ such that 
		\begin{equation}
			\frac{\partial V}{\partial \omega}\geq 0,\qquad \frac{\partial W_j}{\partial \omega}\geq 0 .
			\label{dVdOm}
		\end{equation}
		Then the conclusion of Property \ref{PropdZdOm} holds:
		$\frac{dZ}{d\omega}<0$, and Theorem \ref{TheoZ} is still valid.
		\label{CorodZdOm}
	\end{Corollary}
	Now we discuss three models involving \eqref{Helmholtz-Potentiel} and \eqref{JCspringmass}. In each case, the inner product \eqref{ProdScal} needs to be modified to ensure the self-adjointness of the Hamiltonian. This topic is not pursued here, but references are given to the Reader.
	
	\paragraph{Model 1: imperfect contacts}As a first example, one considers the bulk potential $V(x,\omega)=\rho(x)\,\omega^2$, with the interface parameters and the interface potentials
	\begin{equation}
		\alpha_j=1/K_j,\qquad W_j(\omega)=M_j\,\omega^2.
		\label{SMparameters}
	\end{equation}
	In \eqref{SMparameters}, $K_j>0$ and $M_j\geq 0$ are stiffness and mass terms. The mirror symmetry implies that $K_j=K_{N-j+1}$ $(j=1,\cdots,N$) and similarly for $M_j$. This model describes imperfect transmission of elastic waves through glue layers or cracks \cite{Assier20}. Conservation of energy is proven in \cite{Bellis21}. The usual case of perfect contact is recovered when $K_j\rightarrow +\infty$ and $M_j=0$. The inner product is defined in Appendix A of \cite{Assier20}. Since $\partial_\omega W_j=2\,M_j\,\omega \geq 0$, the assumptions of Corollary \ref{CorodZdOm} are satisfied. 
	
	\paragraph{Model 2: Helmholtz resonators}A second example concerns the propagation of acoustic waves in a waveguide connected with an array of Helmholtz resonators. This configuration is modelled by the Helmholtz equation \eqref{Helmholtz} and the jump conditions
	\begin{equation}
		\jump{p}_{x_j}=0,\hspace{1cm}\jump{v}_{x_j}=\mathrm{i}\,g_j\,\frac{\omega^2}{\omega^2-\omega_j^2}\,p(x_j),
		\label{JC-HR}
	\end{equation}
	where $v=-p'/(\mathrm{i}\,\omega\rho)$ is the acoustic velocity. The resonance frequencies $\omega_j$ and the coupling coefficients $g_j>0$ are related with the geometry of the $j$th Helmholtz resonator. The jump conditions \eqref{JC-HR} can be replaced by \eqref{JCspringmass} with
	\begin{equation}
		\alpha_j=0,\qquad W_j=-g_j\,\frac{\omega^2}{\omega^2-\omega_j^2}.
		\label{V-HR}
	\end{equation}
	Since
	\begin{equation}
		\partial_\omega W_j=2\,g_j\,\omega_j^2\,\frac{\omega}{\left(\omega^2-\omega_j^2\right)^2}>0,
		\label{dVdOm-HR}
	\end{equation}
	the assumptions of Corollary \ref{CorodZdOm} are satisfied. This argument can be used to prove rigorously the existence of topologically protected interface mode in a guide connected with a one-dimensional array of Helmholtz resonators \cite{Zhao21}.
	
	\paragraph{Model 3: dispersive media}As a third and last example, we consider dispersive media with null jump conditions ($\alpha_j=0$ and $W_j=0$ in \eqref{JCspringmass}) but with frequency-dependent compressibility:
	\begin{equation}
		\kappa^{-1}(x,\omega)=\kappa^{-1}_0(x)\left(1-\frac{\Omega_\kappa^2}{\omega^2-\omega_\kappa^2}\right),
		\label{Drude}
	\end{equation}
	with $\kappa_0(x)>0$. Such a parameter is generally obtained through an homogenization process. See \cite{Bellis19} and references therein for application of \eqref{Drude} to acoustics in a guide with an array of Helmholtz resonators in the low frequency range. In photonics, similar expressions hold for the permittivity and / or the permeability in the Drude-Lorentz model. The adequate  inner product is defined in \cite{Cassier17}. Injected in the acoustic Helmholtz equation, the parameter \eqref{Drude} leads to the potential \eqref{Helmholtz-Potentiel} with
	\begin{equation}
		V(x,\omega)=\kappa^{-1}(x,\omega)\,\omega^2.
		\label{Potentiel-Disperse}
	\end{equation}
	Since
	\begin{equation}
		\partial_\omega V=2\,\kappa_0^{-1}(x)\,\omega\left(1+\left(\frac{\Omega_\kappa\,\omega_\kappa}{\omega^2-\omega_\kappa^2}\right)^2\right)>0,
		\label{dVdOm-Drude}
	\end{equation}
	the assumptions of Corollary \ref{CorodZdOm} are satisfied.
	
	
	\section{Numerical examples}\label{SecNum}
	
	\subsection{Configuration}\label{SecNumConfig}
	
	\begin{figure}[htbp]
		\begin{center}
			\begin{tabular}{c}
				\includegraphics[scale=0.45]{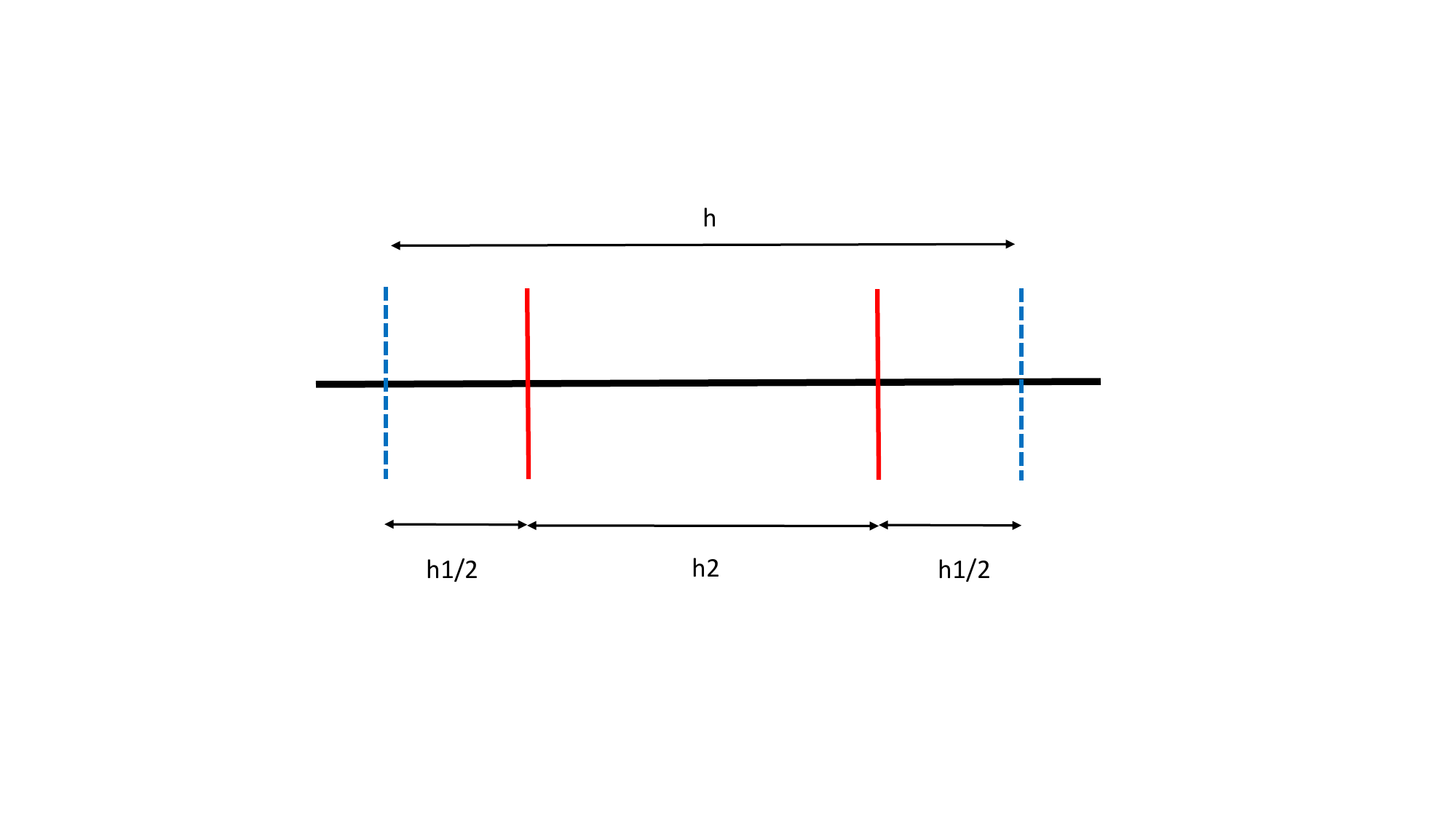} 
			\end{tabular}
		\end{center}
		\vspace{-2.6cm}
		\caption{\label{FigCentro}Elementary cell of a mirror symmetric PC. The red  vertical solid lines denote the two interfaces with imperfect contacts. The blue vertical dotted lines denote the edges of the cell.} 
	\end{figure}
	
	The elementary cell 
	contains two interfaces with imperfect contacts \eqref{JCspringmass}-\eqref{SMparameters} of identical stiffness 
	$Kh/E=5.2$ and $M=0$. The cell is mirror symmetric: the positions of the interfaces in the cell are denoted by $h_1=\theta h$ and $h_2=(1-\theta)h$, with $0<\theta<1$ (Figure \ref{FigCentro}). The transformation $\theta\mapsto1-\theta$ yields a similar PC with identical gaps, contrary to the case of perfect contacts separating different media \cite{Xiao14}. However, the symmetries of the Bloch modes may differ, which will be useful for the construction of topologically protected interface modes. 
	
	\begin{figure}[h!]
		\begin{center}
			\begin{tabular}{cc}
				(a) & (b)\\
				\hspace{-1cm}
				\includegraphics[width=0.55\textwidth]{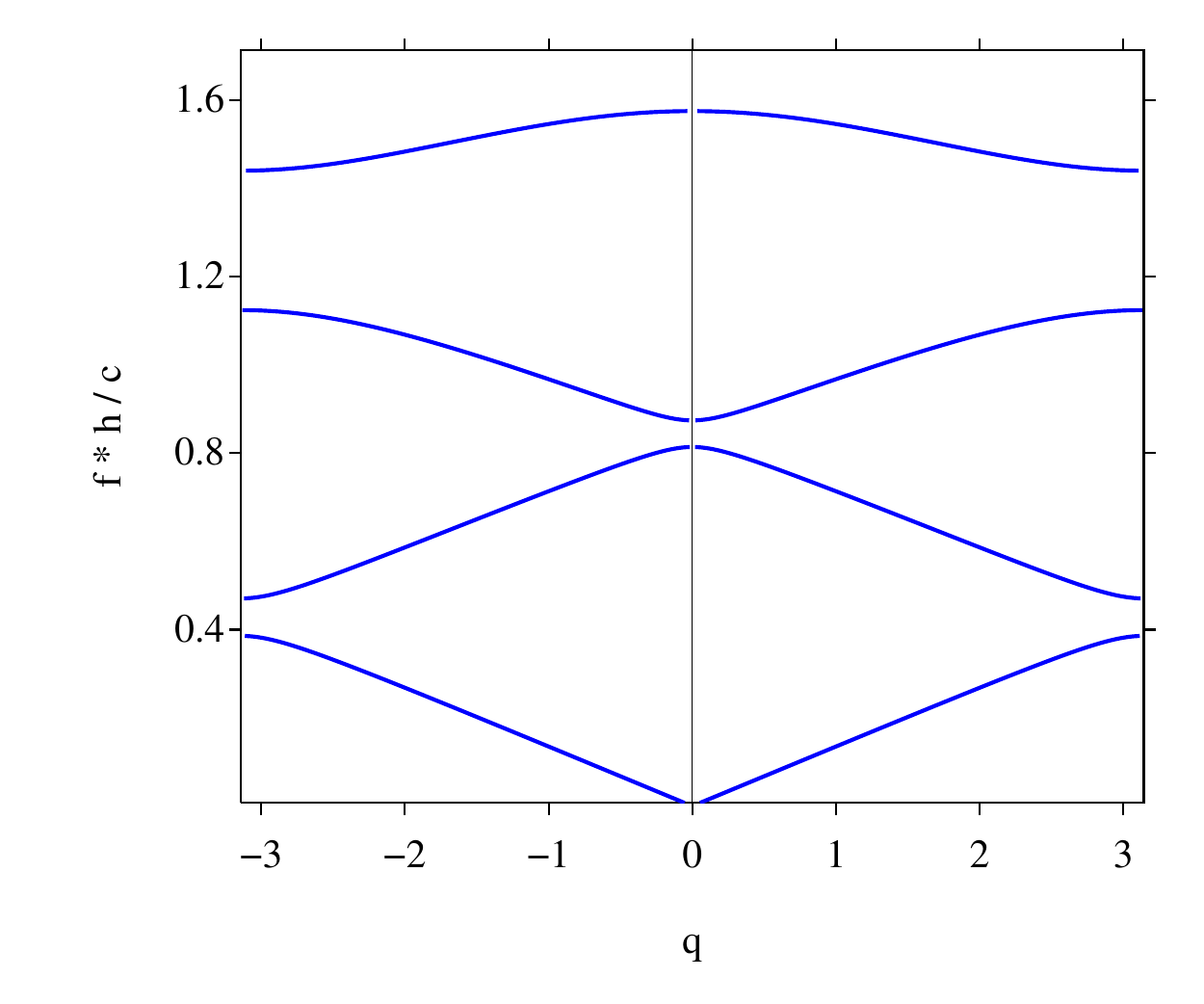} &
				\hspace{-1cm}
				\includegraphics[width=0.55\textwidth]{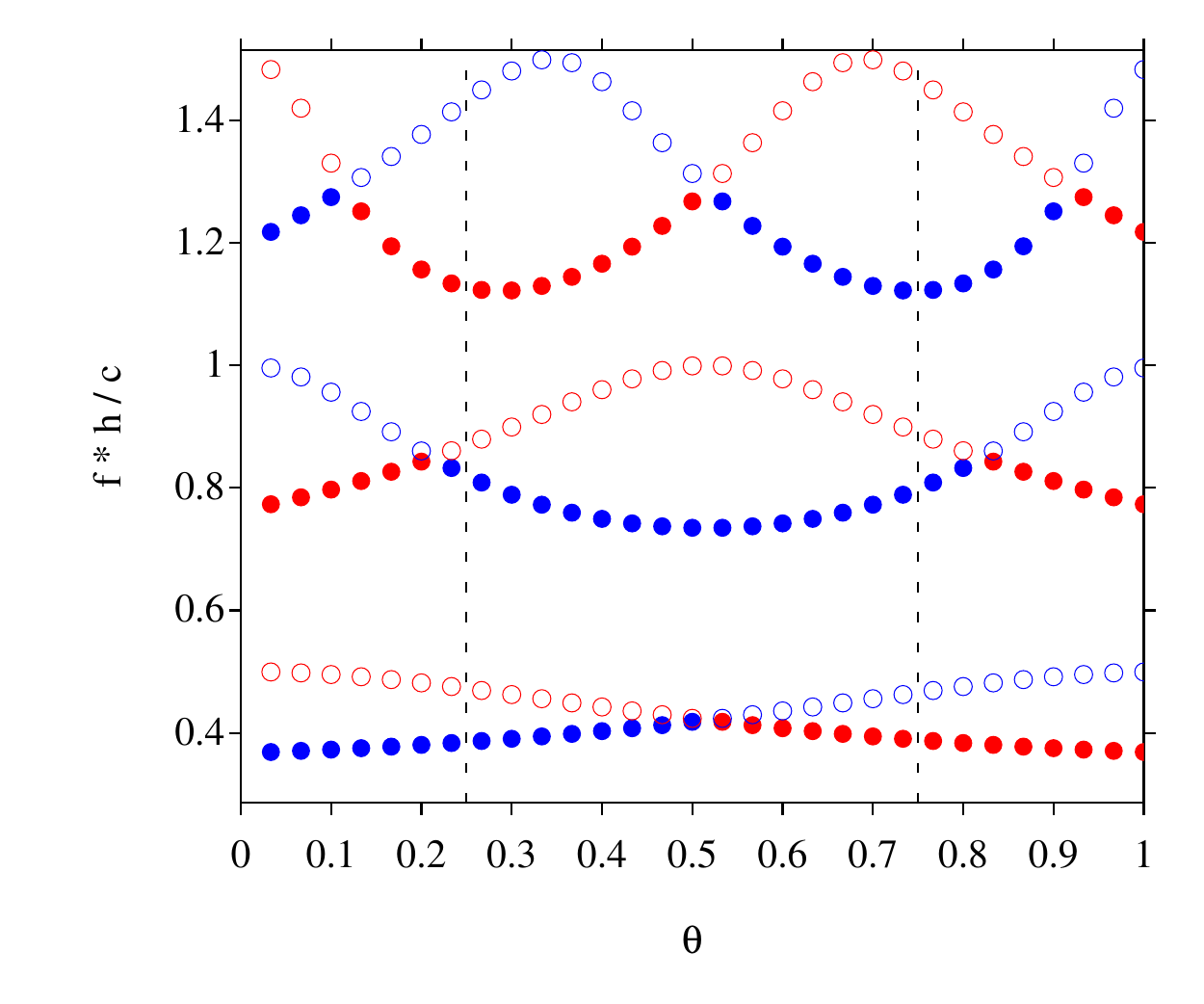} 
			\end{tabular}
		\end{center}
		\vspace{-0.6cm}
		\caption{\label{FigFloquet}(a) Bloch-Floquet dispersion diagram in the case $\theta=0.25$, with the Bloch shift $q=kh$, in scaled frequency (a). Parametric study of the Bloch modes at the band edges, in terms of $\theta$ (b). The full and empty circles denote the lower edge $\varepsilon_n^+$ and the upper edge $\varepsilon_{n+1}^-$ of the gap $\Omega_n$ ($n=1\cdots 3$), respectively. The blue and red circles denote symmetric and antisymmetric Bloch modes, respectively. The vertical dotted lines denote $\theta=0.25$ and $\theta=0.75$.} 
	\end{figure}
	
	Figure \ref{FigFloquet}-(a) displays the Bloch-Floquet dispersion diagram in the case $\theta=0.25$. The vertical axis shows the range of scaled frequencies $\tilde{f}=f\times h / c$. A similar scaling is used all along the text; notably, one denotes $\tilde{f}_n^\pm=\sqrt{\varepsilon_n^\pm}/(2\pi)\times h/c$. As stated in Property \ref{PropCS}, the Bloch modes at these edges are either symmetric or antisymmetric. Figure \ref{FigFloquet}-(b) displays a parametric study of the symmetries of Bloch modes at the band edges, in terms of $\theta$. The lower frequency $\varepsilon_n^+$ and the upper frequency $\varepsilon_{n+1}^-$ of the gaps $\Omega_n$ are denoted by full and empty circles, respectively. A blue circle represents a symmetric Bloch wave, whereas a red circle represents an antisymmetric Bloch wave. One notices the symmetry of the frequencies with respect to $\theta=0.5$, induced by the invariance of the PC by the transformation $\theta\mapsto 1-\theta$. However, the symmetries of Bloch modes may differ when comparing  $\theta$ and $1-\theta$. 
	
	\begin{figure}[h!]
		\begin{center}
			\begin{tabular}{cc}
				(a) & (b)\\
				\hspace{-1cm}
				\includegraphics[width=0.55\textwidth]{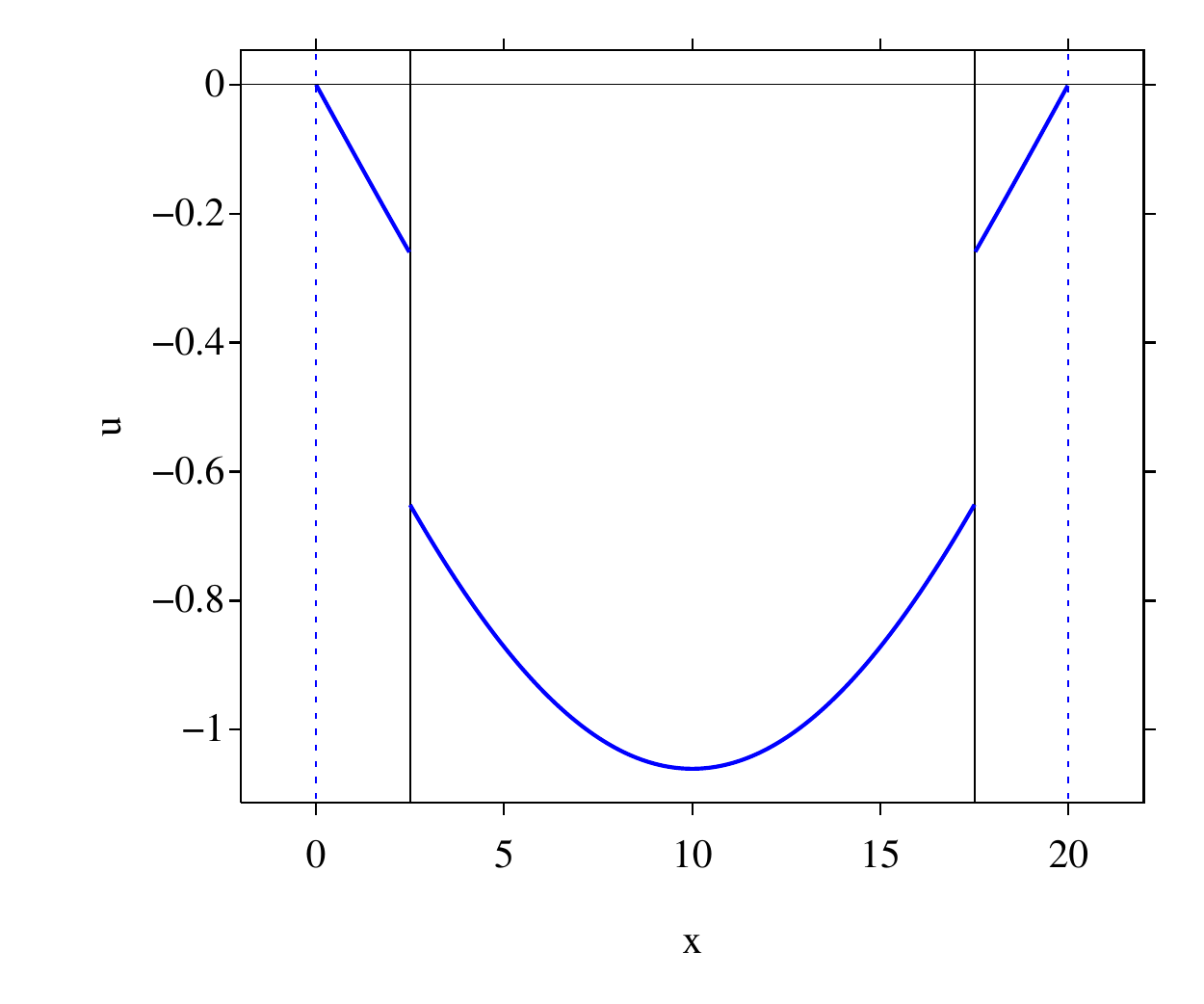} &
				\hspace{-1cm}
				\includegraphics[width=0.55\textwidth]{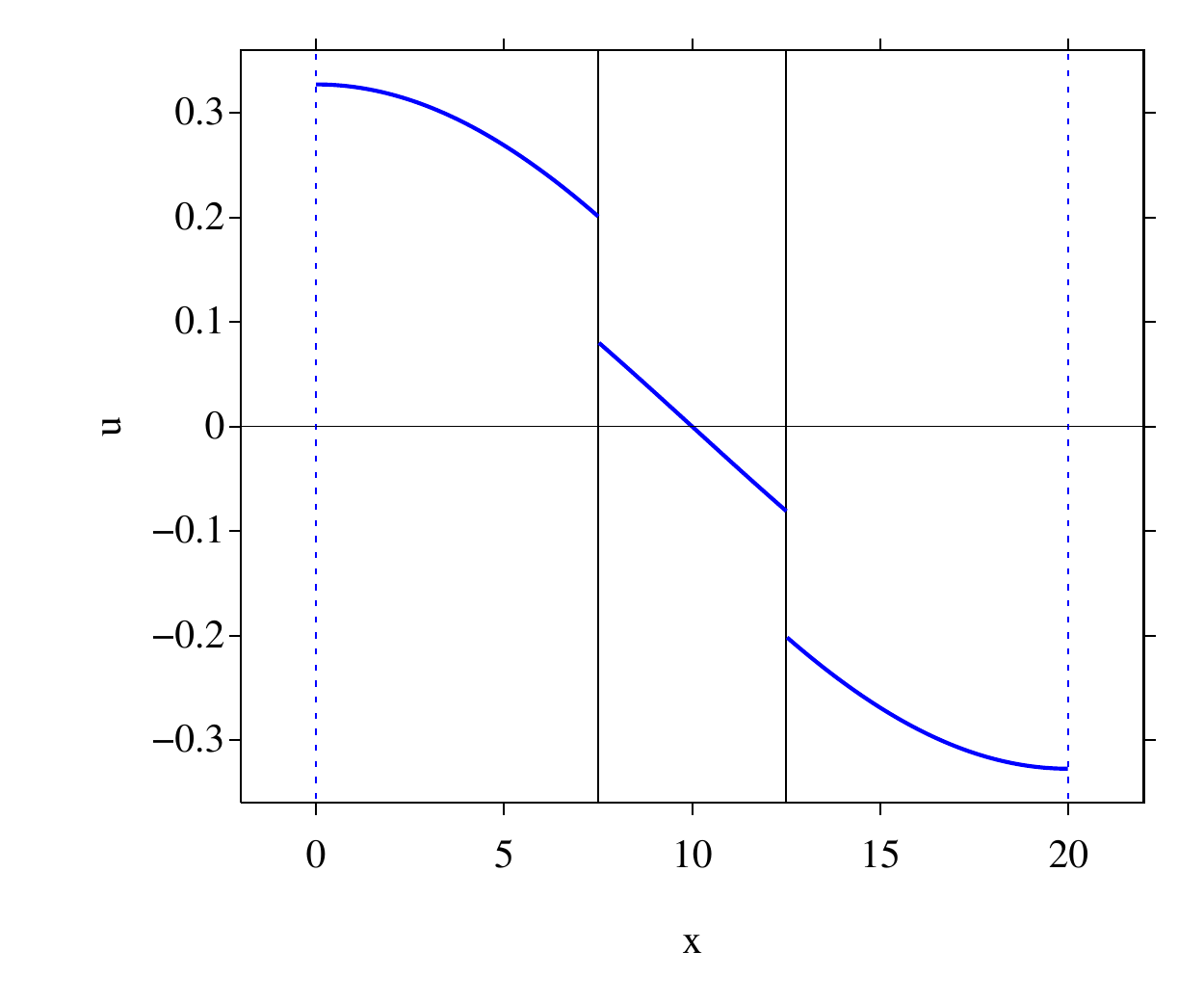}
			\end{tabular}
		\end{center}
		\vspace{-0.6cm}
		\caption{\label{FigModeSym}Bloch wave $u_1$ at the lower edge of gap $\Omega_1$ ($\tilde{f}_1^+=0.385$). (a): $\theta=0.25$, (b): $\theta=0.75$. The vertical solid and dotted lines denote the interfaces and the edges of the elementary cell, respectively.} 
	\end{figure}  
	
	From now on, we focus on two particular values of $\theta$: 0.25 and 0.75. The PCs built using these values are named PC-L and PC-R, respectively. The intersections between the bands and the vertical dotted lines in Figure \ref{FigFloquet}-(b) determine the symmetries of $u_n$ at the band edges. For instance, $u_1$ at the lower edge $\tilde{f}_1^+=0.385$ of the gap $\Omega_1$ is expected to be symmetric (at $\theta=0.25$) and antisymmetric (at $\theta=0.75$). Figure \ref{FigModeSym} displays these Bloch modes; Figure \ref{FigModeSym}-(a) amounts to $u_1$ in Figure \ref{FigBloch}-(a). One observes the change of symmetry of $u_1$ between PC-L and PC-R at this band edge.
	
	\begin{table}[h!]
		\begin{center}
			\begin{small}
				\begin{tabular}{c||c|c|c|c}
					& $n=1$ & $n=2$ & $n=3$ & $n=4$ \\
					\hline 
					$[\tilde{f}_n^+,\tilde{f}_{n+1}^-]$ & $[0.385,\,0.470]$ & $[0.814,\,0.874]$ & $[1.124,\,1.440]$     & $[1.575,\,2.000]$ \\
					PC-L & ${\cal S}$ & ${\cal S}$ & ${\cal A}$ & ${\cal S}$ \\
					PC-R & ${\cal A}$ & ${\cal S}$ & ${\cal S}$ & ${\cal S}$ \\ 
				\end{tabular}
			\end{small} 
		\end{center}
		\vspace{-0.5cm}
		\caption{\label{TabBG}Scaled frequency intervals of the gaps $\Omega_n$ ($n=1\cdots 4$). Symmetries of the Bloch modes at scaled frequencies $\tilde{f}_n^+$, in PC-L ($\theta=0.25$) and PC-R ($\theta=0.75$).}
	\end{table}
	
	The scaled frequencies and the symmetry of $u_n$ at the band edges of $\Omega_n$ ($n=1\cdots 4$) are given in Table \ref{TabBG}. Based on Theorem \ref{TheoZ}, we expect the existence of topologically protected interface modes in the gaps $\Omega_1$ and $\Omega_3$. On the contrary, no interface modes are expected in the gaps $\Omega_2$ and $\Omega_4$, where the symmetries of $u_n$ are identical when $\theta=0.25$ and $\theta=0.75$. Now we glue PC-L and PC-R at $x=0$ to illustrate this claim.
	
	\begin{figure}[h!]
		\begin{center}
			\begin{tabular}{cc}
				(a) $\Omega_1$ & (b) $\Omega_2$ \\
				\hspace{-1cm}
				\includegraphics[width=0.55\textwidth]{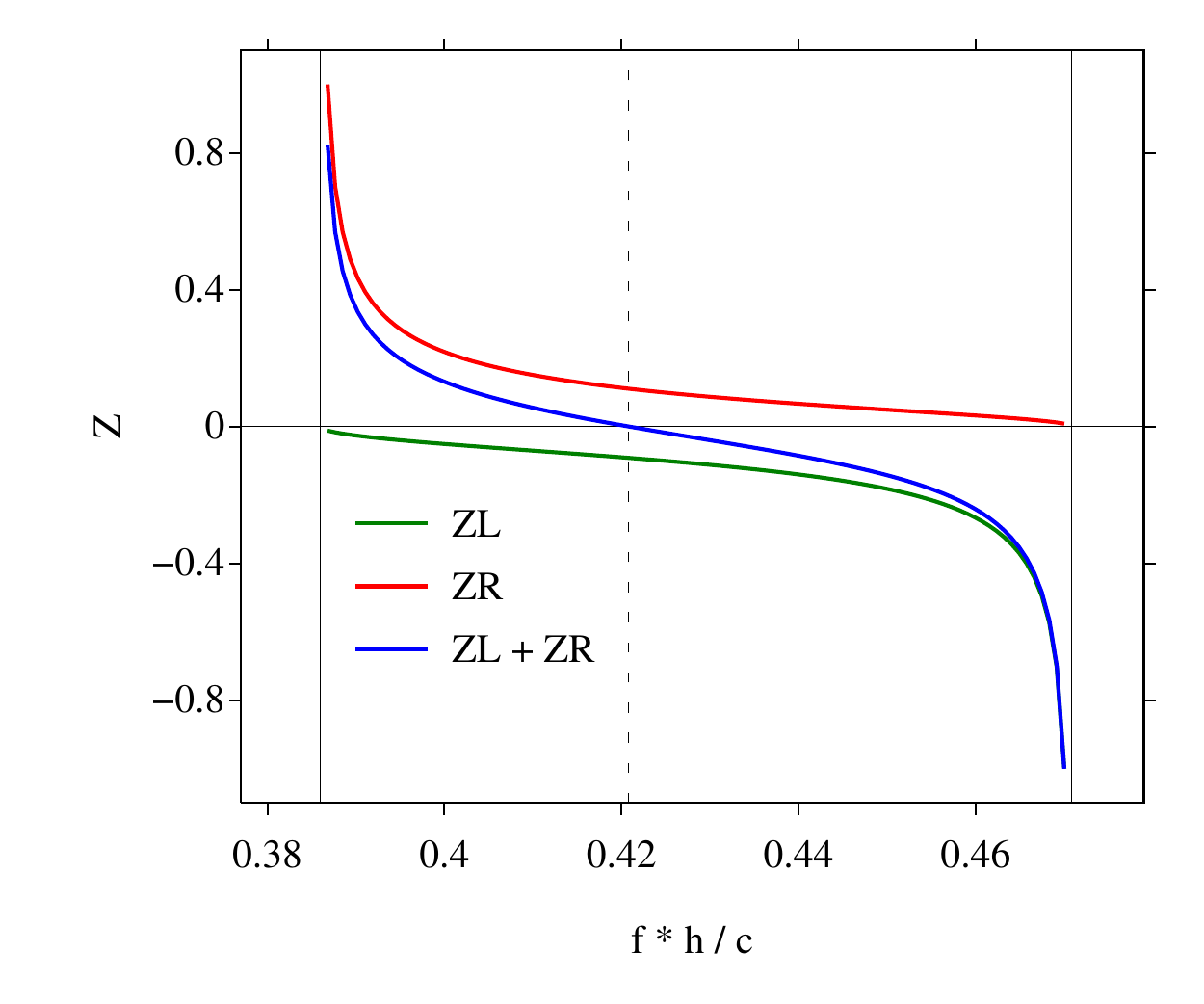} &
				\hspace{-1cm}
				\includegraphics[width=0.55\textwidth]{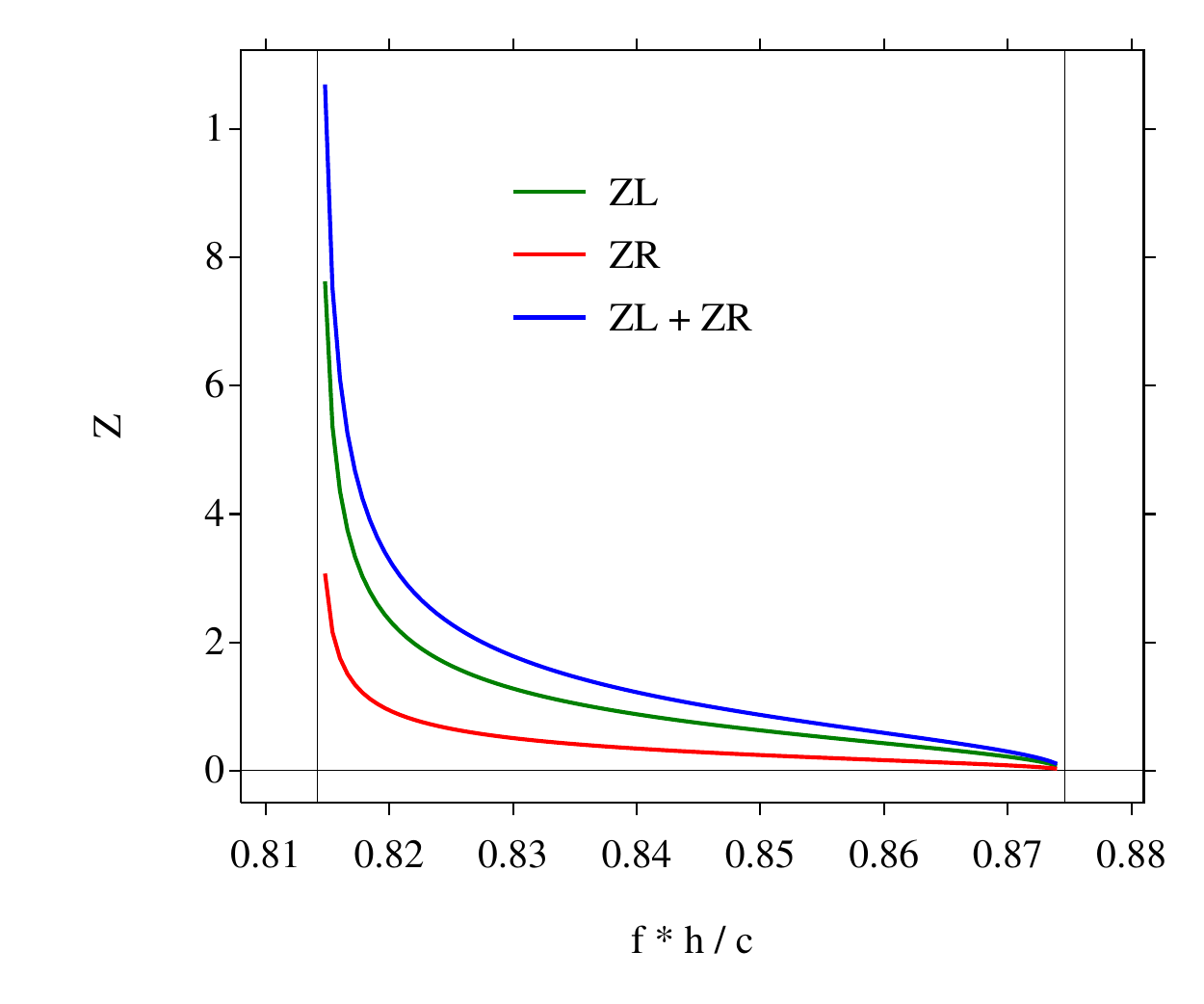} \\
				(c) $\Omega_3$ & (d) $\Omega_4$ \\
				\hspace{-1cm}
				\includegraphics[width=0.55\textwidth]{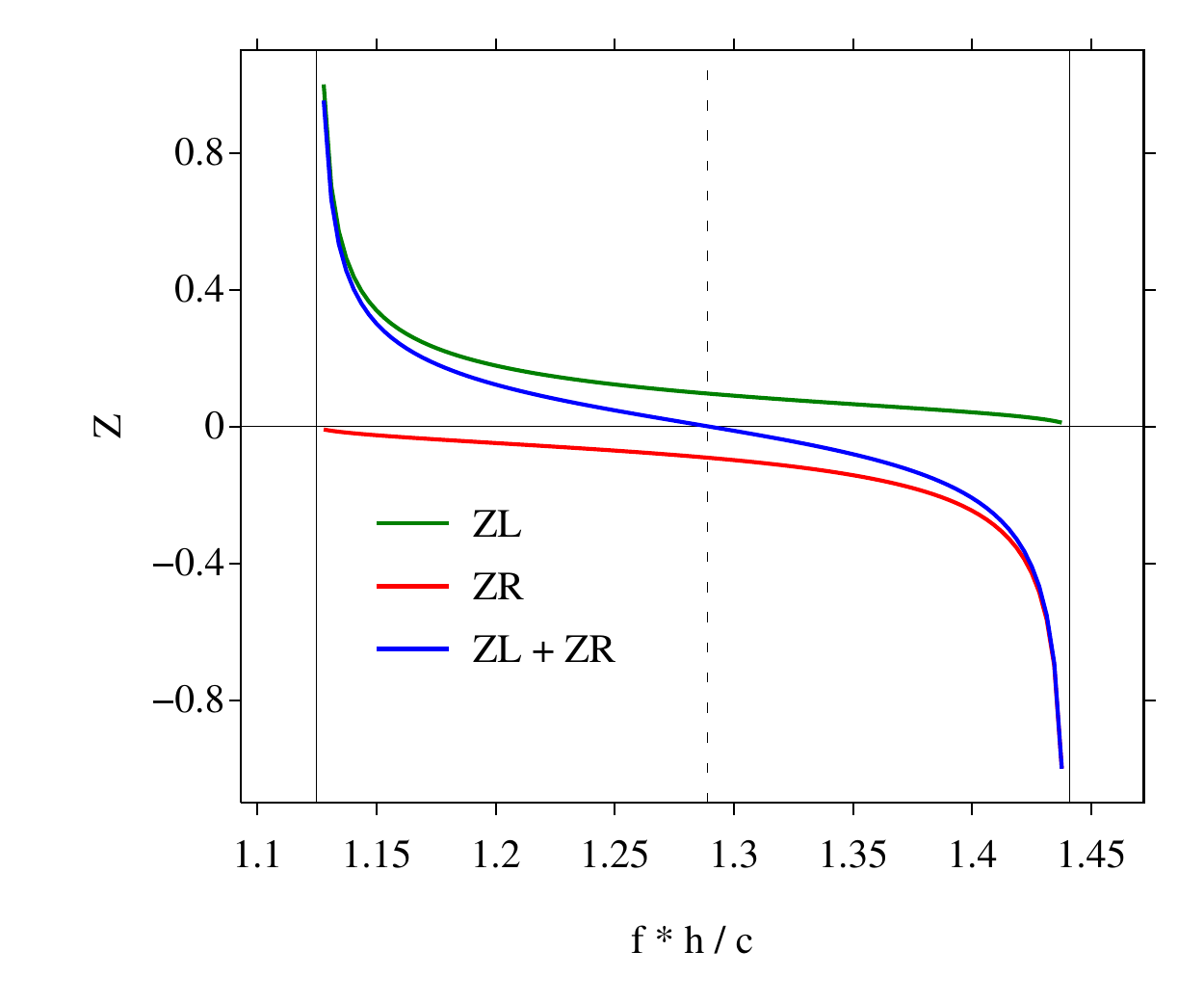} &
				\hspace{-1cm}
				\includegraphics[width=0.55\textwidth]{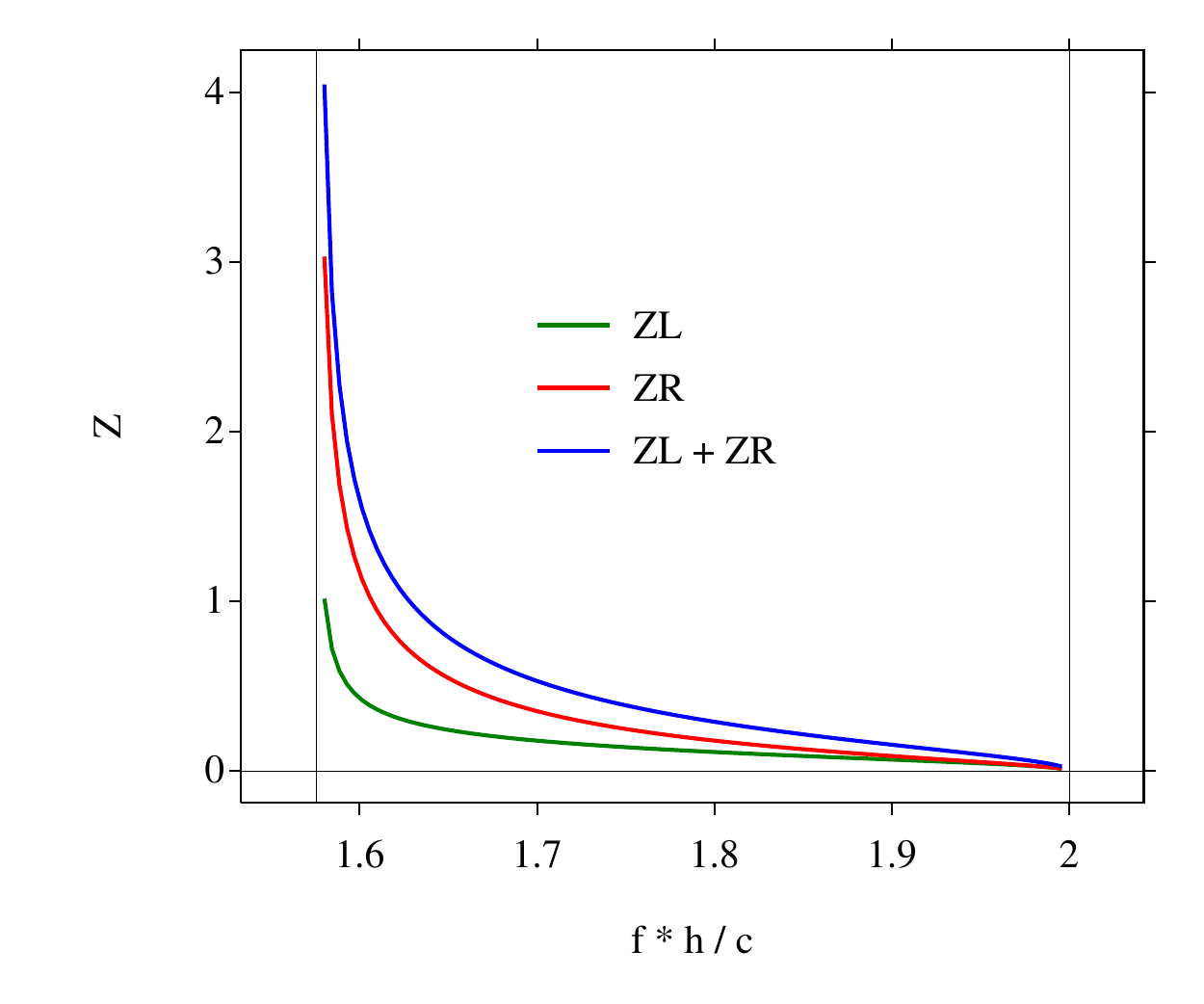}
			\end{tabular}
		\end{center}
		\vspace{-0.6cm}
		\caption{\label{FigZ}Frequency evolution of the surface impedances of the PCs in contact. The vertical solid lines denote the edges of the $n$th gap in scaled frequency. (a-c): gaps $\Omega_1$ and $\Omega_3$, in which a topologically protected interface mode exists. The vertical dotted line denotes the scaled frequency $\tilde{f}_1^\sharp\approx 0.42$ and $\tilde{f}_3^\sharp \approx 1.28$ where $Z_L+Z_R=0$. (b-d): gaps $\Omega_2$ and $\Omega_4$, in which no topologically protected interface mode exists.} 
	\end{figure}
	
	Figure \ref{FigZ} shows the frequency evolution of the surface impedance $Z_L$ in PC-L, of the surface impedance $Z_R$ in PC-R, and finally of their sum. In Figure \ref{FigZ}-(a), we observe that $Z_L+Z_R=0$ at $\tilde{f}_1^\sharp\approx 0.420$  in the gap $\Omega_1$: according to Theorem \ref{TheoZ}, a topologically protected interface mode exists at this frequency. In Figure \ref{FigZ}-(b,d), $Z_L+Z_R \neq 0$: no interface mode exists. In Figure \ref{FigZ}-(c), $Z_L+Z_R=0$ at $\tilde{f}_3^\sharp\approx 1.289$ in the gap $\Omega_3$: once again, a topologically protected interface mode exists at this frequency.
	
	\subsection{Finite slabs}\label{SecNumSlab}
	
	Now, we look for the expected interface modes in a scattering configuration~\cite{Kalozoumis18}. For this purpose, one considers finite PCs consisting of $N$ cells with $\theta=0.25$ on the left ($x<0$), and $N$ cells with $\theta=0.75$ on the right ($x>0$). Figure \ref{FigTra} shows the frequency evolution of the transmission coefficient through the slab of $2N$ cells, for $N=3$ (a) and $N=5$ (b). Successions of $N-1$ equidistant oscillations are observed in the bands, corresponding to the modes of a cavity. More interestingly, isolated peaks are observed in the gaps $\Omega_1$ (at $\tilde{f}\approx 0.420$) and $\Omega_3$ (at $\tilde{f}=1.289$). These scaled frequencies corresponds very accurately to the roots $\tilde{f}_1^\sharp$ and $\tilde{f}_3^\sharp$ of $Z_L+Z_R=0$ observed in Figure \ref{FigZ}-(a,c). As $N$ increases, these peaks become thinner. In \ref{FigTra}-(b), the peak in $\Omega_3$ is so thin that the frequency discretization is insufficient to capture its spatial support.
	
	\begin{figure}[h!]
		\begin{center}
			\begin{tabular}{cc}
				(a) & (b)\\
				\hspace{-1cm}
				\includegraphics[width=0.55\textwidth]{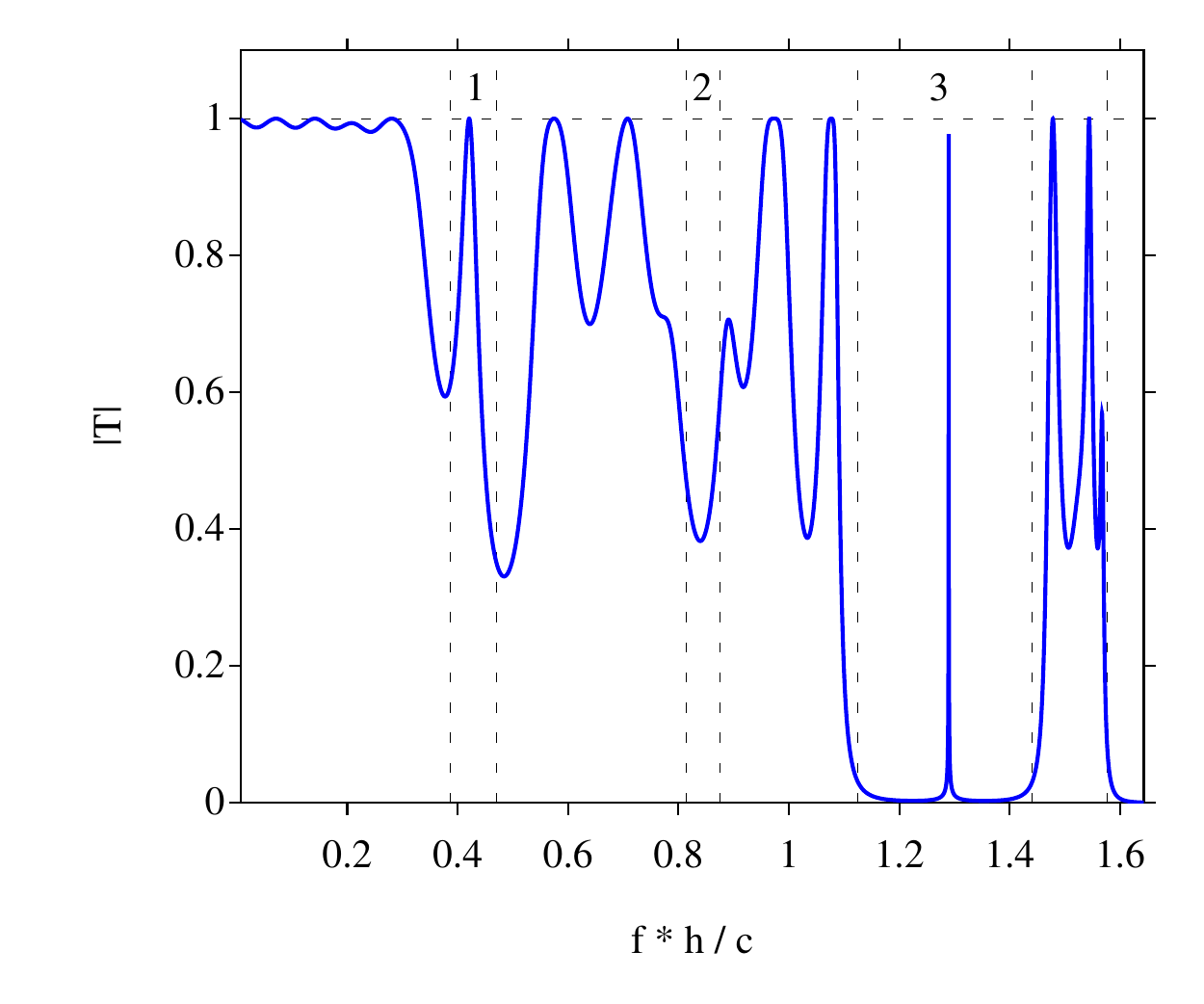} &
				\hspace{-1cm}
				\includegraphics[width=0.55\textwidth]{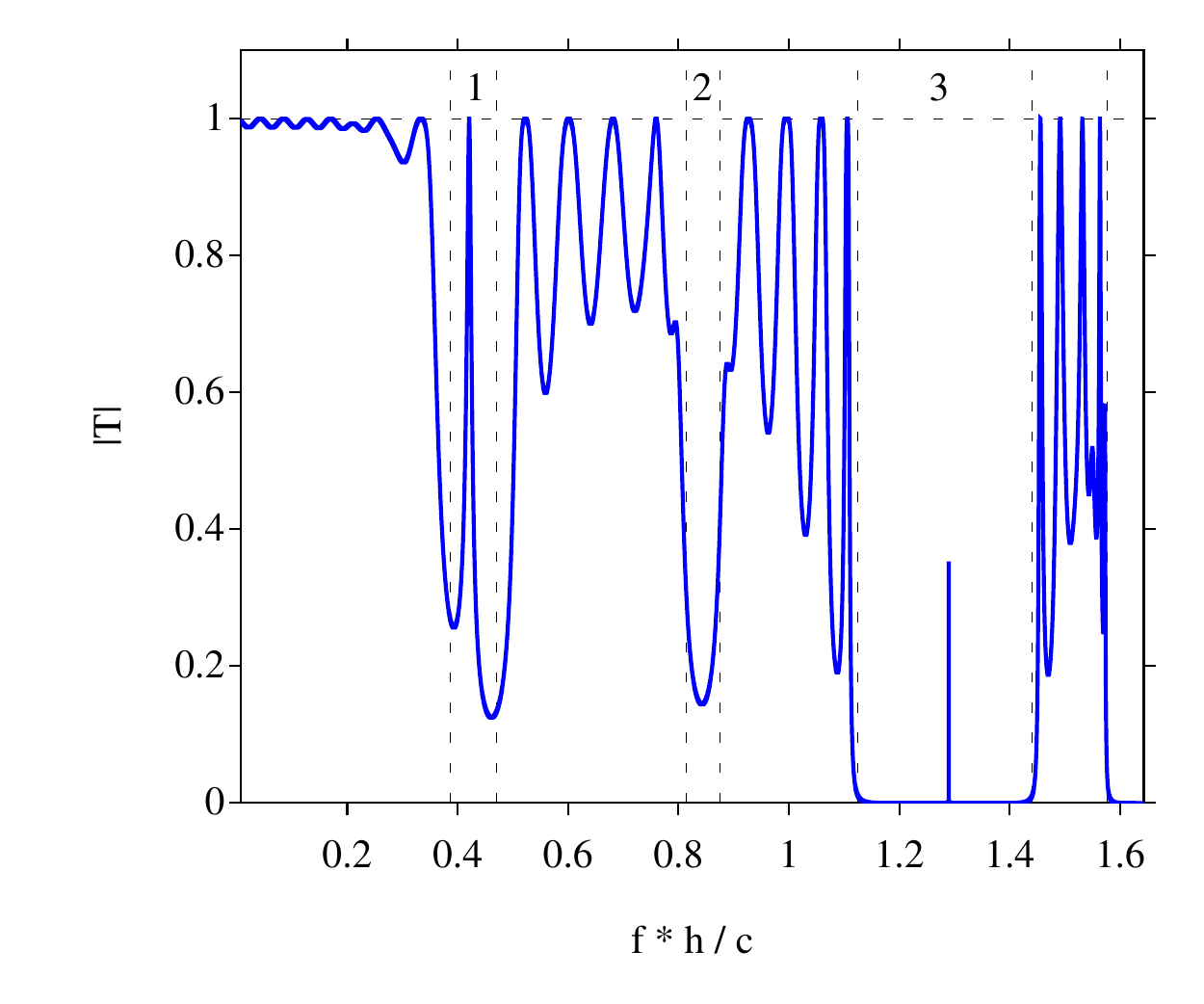} 
			\end{tabular}
		\end{center}
		\vspace{-0.6cm}
		\caption{\label{FigTra}Modulus of the transmission coefficient in the case of $N=3$ cells (a) and $N=5$ cells (b) per PC. The vertical dashed lines denote the lower and upper edges of the gaps $\Omega_n$ ($n=1\cdots 3$), denoted by their numbering. An isolated peak is observed in the gaps $\Omega_1$ and $\Omega_3$.} 
	\end{figure}
	
	Figure \ref{FigModeU} represents the spatial evolution of the modulus of $u$ at $\tilde{f}_1^\sharp$ and $\tilde{f}_3^\sharp$, with $N=3$ cells in PC-L and PC-R. In both cases, an evanescent mode centered on the interface between the PCs is observed. It is a clear signature of an interface mode at the interface between PC-L and PC-R.
	
	\begin{figure}[h!]
		\begin{center}
			\begin{tabular}{cc}
				(a) & (b)\\
				\hspace{-1cm}
				\includegraphics[width=0.55\textwidth]{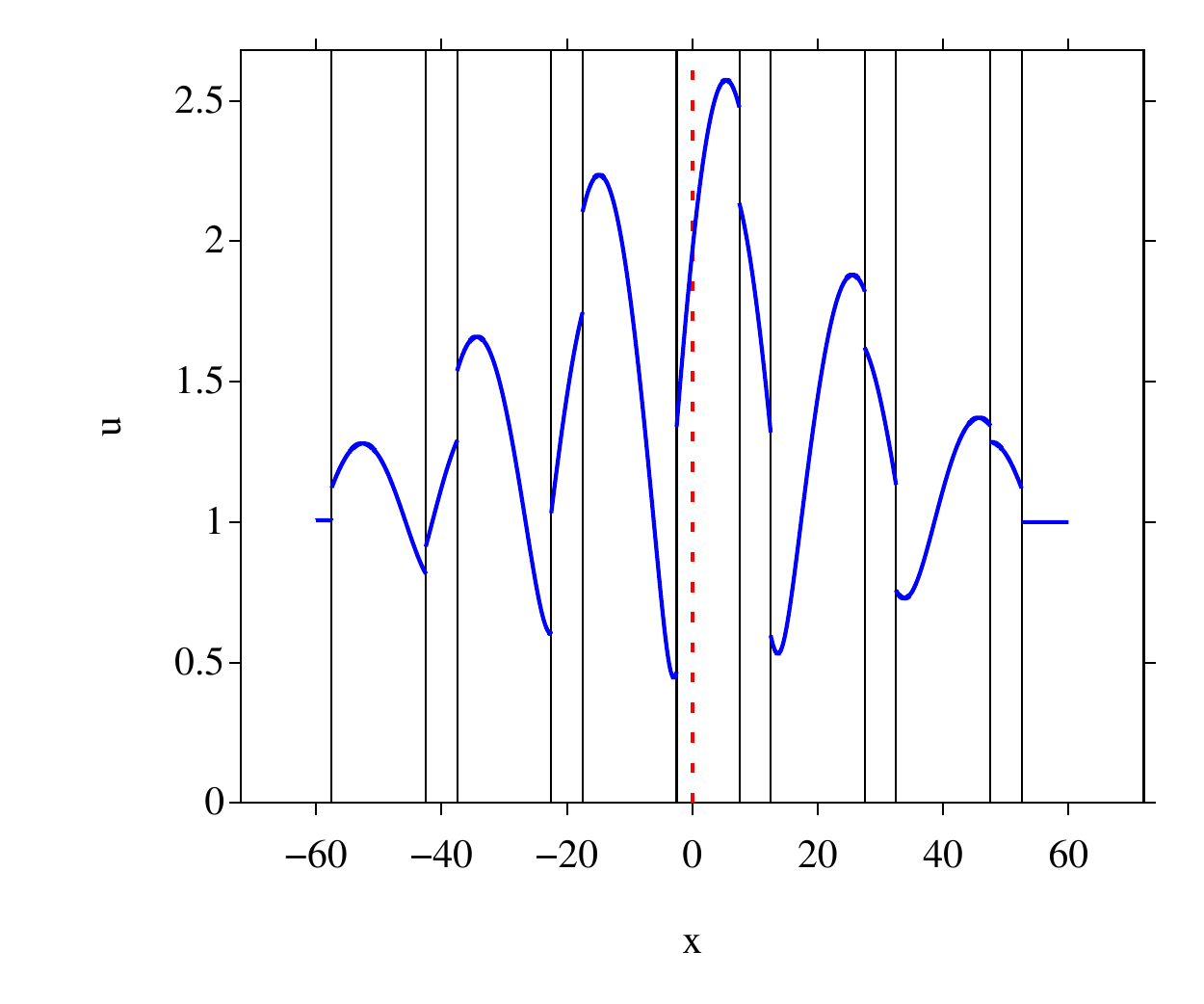} &
				\hspace{-1cm}
				\includegraphics[width=0.55\textwidth]{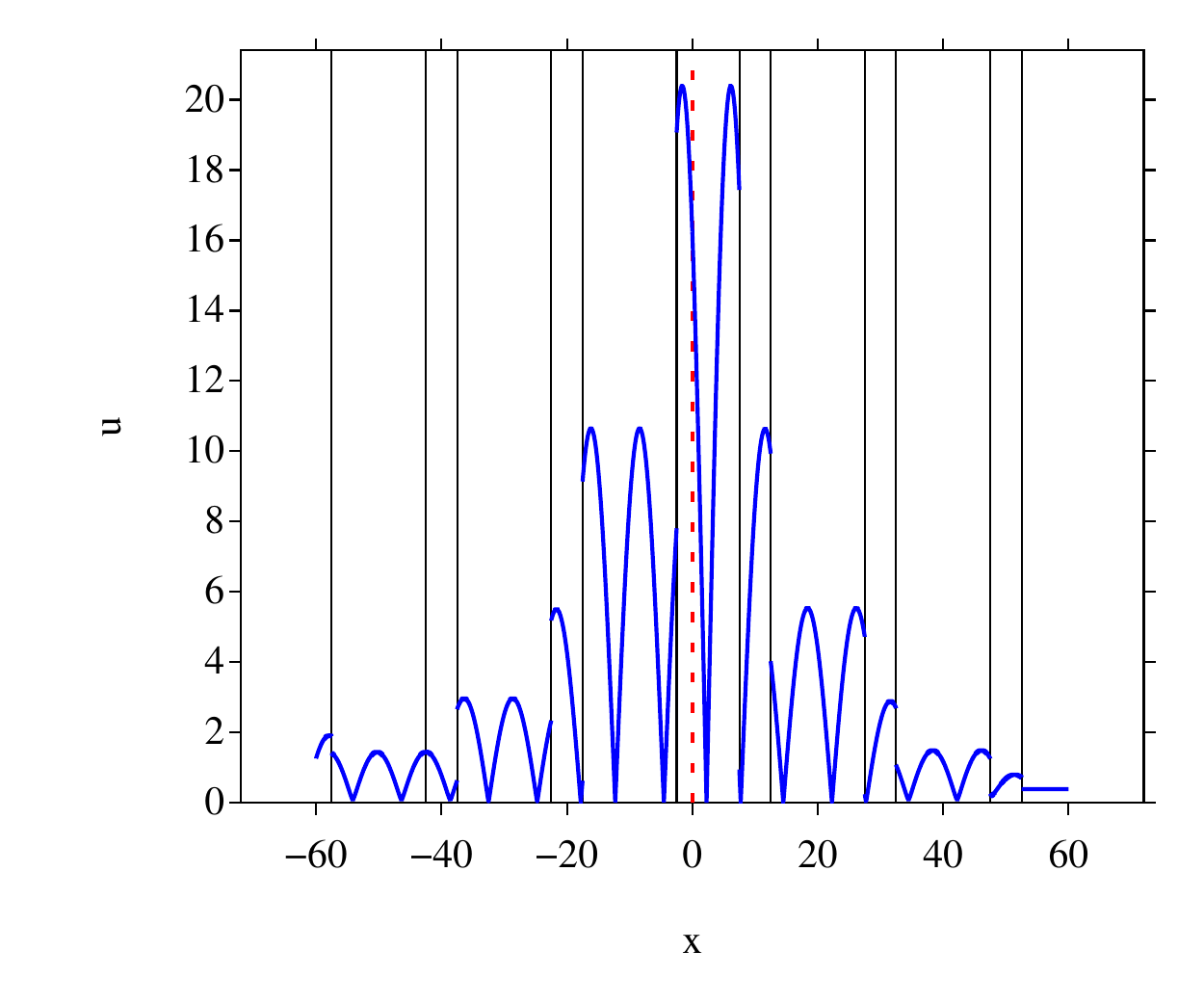} 
			\end{tabular}
		\end{center}
		\vspace{-0.6cm}
		\caption{\label{FigModeU}Spatial evolution of $u$ at the scaled frequencies $\tilde{f}_1^\sharp=0.420$ (a) and $\tilde{f}_3^\sharp=1.289$. The vertical solid lines represent the interfaces with imperfect contacts. The red vertical dotted line at $x=0$ denotes the interface between PC-L and PC-R, each being built with $N=3$ cells.} 
	\end{figure}
	
	Finally, we study the evolution of the interface modes when the geometry of the elementary cell varies. The length $h$ remains constant, but the parameter $\theta$ varies; the PC-R is built using the parameter $1-\theta$. Figure \ref{FigZTheta} shows $\tilde{f}_1^\sharp$ (a) and $\tilde{f}_3^\sharp$ (b) as functions of $\theta$. In (b), the minimum value of $\theta$ is 0.1; below this value, a Dirac point exists, as observed in the gap $\Omega_3$ on Figure \ref{FigBloch}-(b). The blue lines denote the lower and upper edges of the gap $\Omega_1$ (a) and $\Omega_3$ (b). The red line denote the scaled frequency of the interface mode, computed as the zero of $Z_L+Z_R$. At the scale of the Figure, this frequency seems contant when $\theta$ varies, which shows the robustness of the topologically protected interface modes. Nevertheless, a large zoom on the red line would show that $\tilde{f}_3^\sharp$ is not rigourously constant.
	
	\begin{figure}[h!]
		\begin{center}
			\begin{tabular}{cc}
				(a) & (b)\\
				\hspace{-1cm}
				\includegraphics[width=0.55\textwidth]{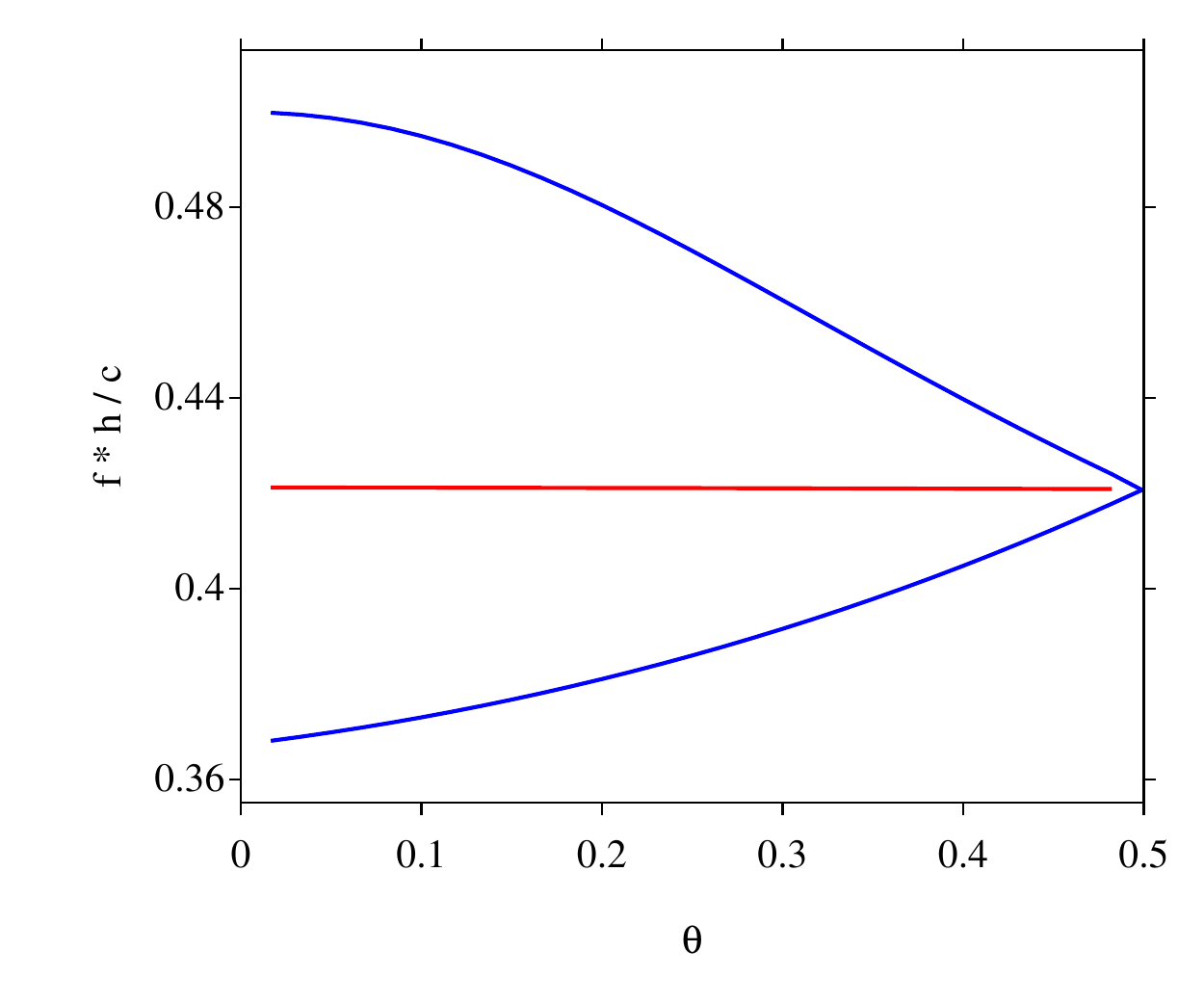} &
				\hspace{-1cm}
				\includegraphics[width=0.55\textwidth]{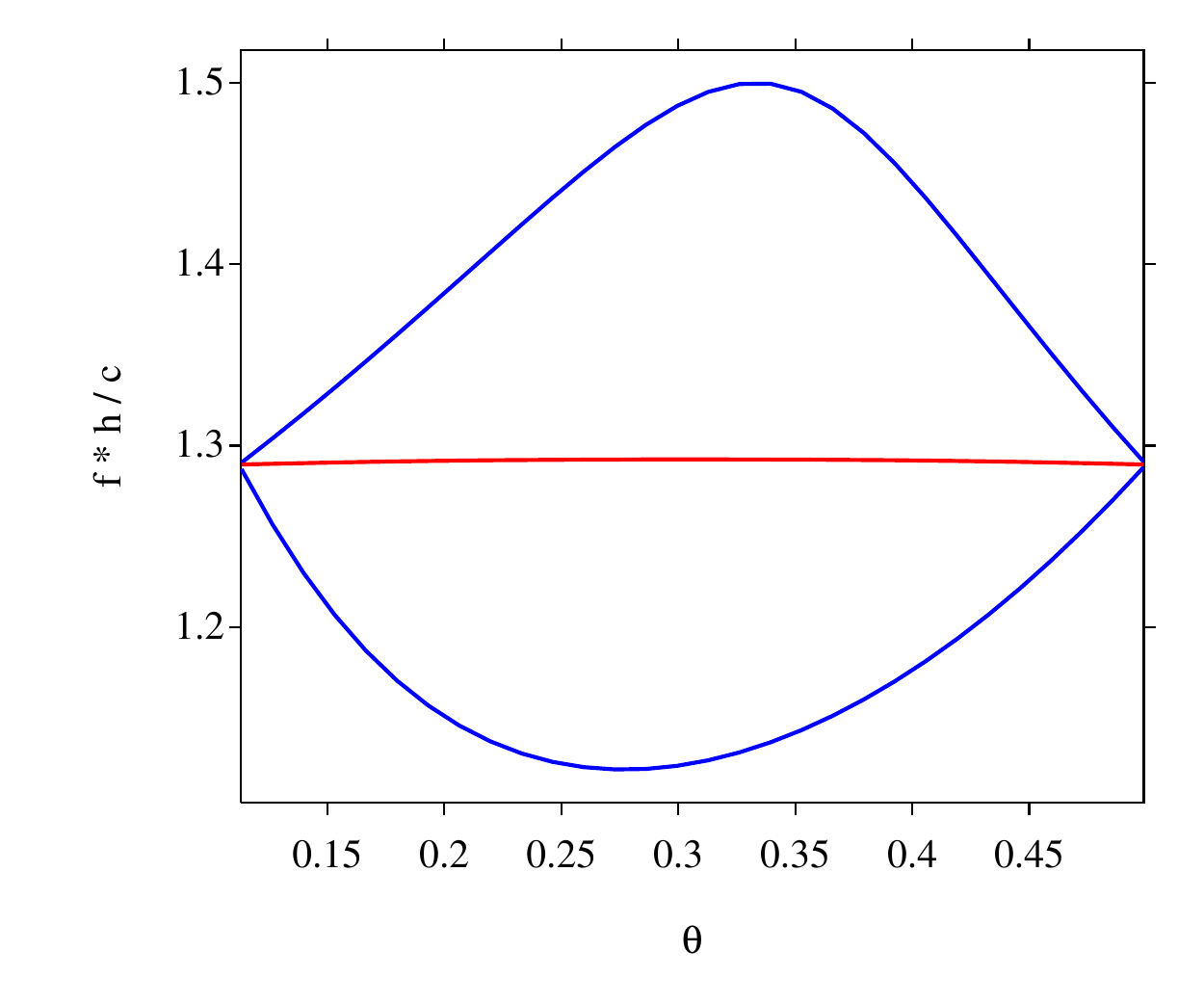} 
			\end{tabular}
		\end{center}
		\vspace{-0.6cm}
		\caption{\label{FigZTheta}Frequency evolution of the topologically protected interface modes in the gaps $\Omega_1$ (a) and $\Omega_3$ (b), as a function of the parameter $\theta$ governing the geometry of the elementary cell. These scaled frequencies $\tilde{f}_1^\sharp$ (a) and $\tilde{f}_3^\sharp$ (b) are represented by a red line. The blue curves represent the scaled lower edge $\tilde{f}_n^+$ and scaled upper upper $\tilde{f}_{n+1}^-$ of the gaps $\Omega_n$.} 
	\end{figure}
	
	
	\section{Conclusion}\label{SecConclu}
	
	In this work, we considered semi-infinite PCs with mirror symmetry, which are joined together. To prove the existence of topologically protected interface modes in their common gaps, we used the concept of surface impedance. Initially introduced in \cite{Xiao14} in the case of a bilayer medium, we extended this approach to any mirror symmetric PC in 1D. The main result obtained in Theorem \ref{TheoZ} states that a change of Zak phase, or equivalently a symmetry inversion of Bloch modes at gap edges, guarantees the existence and unicity of topologically protected interface localized modes. 
	
	This work opens several research directions. The first natural extension would be the generalization of this approach to higher spatial dimensions. Zak phases have been used as topological indices in 2D systems characterizing the presence of edge waves~\cite{Liu17,Liu17b}, but a relevant bulk-boundary correspondence has not been established so far~\cite{Xu23}. Similarly, the concept of higher-order topological insulators have attracted a lot of attention in the recent years~\cite{Schindler18,Benalcazar17}. For instance, 2D higher order topological insulators can host localized modes in their corners. Again, a complete understanding of the a higher-order bulk-boundary correspondence is still lacking, but the similarity with the present problem suggests that similar techniques might lead to progresses in that direction. Indeed, the study of 2D photonic surface modes from surface impedances was carried out for instance in \cite{Lawrence10}, but so far without the prism of topology. 	
	
	A simpler generalization of the present work would be to extend the proof to quasi-1D regimes. This work and previous ones~\cite{Fefferman14,Xiao14,Gontier20,Drouot21,Lin22} heavily rely on having only two modes at a given frequency (equation~\eqref{Helmholtz} is second order in space). It is an open question how to extend this to multimodal systems, such as finite width waveguide~\cite{Pagneux96}, or dispersive media containing higher-order spatial derivatives, such as flexural beams \cite{Carta15}. 
	
	Another interesting research direction would be the investigation of topological modes in nonlinear regimes. In \cite{Chaunsali19}, the emergence of a topological mode in a discrete system alternating different nonlinear springs is shown, depending on the amplitude of the perturbations. The study of nonlinear topological modes in a continuous medium remains a  open subject. Let us note that the imperfect conditions \eqref{JCspringmass} easily allow the introduction of nonlinear mechanisms \cite{Bellis21}.

	\appendix
	
	\section{Zak phase as a topological invariant}\label{Sec-ZakTopo}

	In this appendix, we show how to relate the Zak phase with the symmetries of the Bloch modes on the band edges. To ease the discussion, we start by recalling the definition of the Berry connection in equation~\eqref{Berry}: 
	\begin{equation}
		\label{BerryApp}
		A_n(q)=-\mathrm{i} \mean{u_n(q)|\partial_q u_n(q)}.  
	\end{equation}
	The Berry connection in~\eqref{BerryApp} is defined up to a gauge transformation. Indeed, one can always redefine the Bloch modes by changing its phase, as 
	\begin{equation}
		\tilde u_n(q) = e^{\mathrm{i} \theta(q)} u_n(q). 
	\end{equation}
	Assuming that $\theta$ is a smooth function of $q$, this gives the new connection
	\begin{equation}
		\label{Gauge}
		\tilde A_n(q) = A_n(q) + \frac{d \theta}{d q} . 
	\end{equation}
	Remembering that the Zak phase is obtained as the integral of the Berry connection over the Brillouin zone, as defined in equation~\eqref{Zak}, the latter can be computed with either connection. The gauge transformation~\eqref{Gauge} shows that: 
	\begin{equation}
		\int_{-\pi}^{\pi} \tilde A_n(q) d q = \int_{-\pi}^{\pi} A_n(q) d q + \theta(\pi) - \theta(0). 
	\end{equation}
	Since $e^{\mathrm{i} \theta(q)}$ must be $2\pi$-periodic, $\theta(\pi) - \theta(0)$ is a multiple of $2\pi$. This shows that the Zak phase modulo $2\pi$ is a gauge invariant quantity (i.e. independent of a particular choice of Bloch mode basis). At this level, it is worth noting that the Berry connection is defined here using Bloch modes satisfying the Bloch condition $u_q(x+h) = e^{\mathrm{i} q} u_q(x)$ (see equation~\eqref{Floquet_Eig}). The advantage is that $u_q$ is periodic in $q$, and hence, the proof of gauge invariance follows naturally, as we just saw. An alternative definition is sometimes found in the literature~\cite{Cayssol21}, which uses the periodic part of $u_q$: 
	\begin{equation}
		\label{PeriodicBloch_App}
		\psi_q(x) = e^{-\mathrm{i} qx/h} u_q(x) .
	\end{equation}
	While  $\psi_q(x)$ is periodic in $x$, it is no longer in $q$, and hence, the gauge invariance of the Zak phase defined in this manner is quite cumbersome~\footnote{In fact, the original paper by Zak uses the latter choice, and hence contains a lengthy discussion about gauge invariance.}. 
	
	Moreover, when the system is mirror symmetric, the connection at $q$ can be related to that at $-q$. Indeed, by taking the derivative of equation~\eqref{ThetaQ} with respect to $q$, and the scalar product with $u_n(-q)$ we obtain: 
	\begin{equation} 
		A_n(-q) = A_n(q) + \partial_q \xi_n . 
	\end{equation}  
	In other words, the Berry connection at $q$ differs from that at $-q$ by a total derivative. Now, if we integrate this relation over half of the Brillouin zone, we obtain 
	\begin{equation} 
		-\int_{-\pi}^0 A_n(q) d q = \int_{0}^{\pi} A_n(q) d q + \xi_n(\pi) - \xi_n(0). 
	\end{equation} 
	Combining the two integrals gives the integral over the whole Brillouin zone, and hence, the identity~\eqref{PhiZak} for the Zak phase. We can also notice that for mirror symmetric systems, the Zak phase defined using the spatially periodic modes $\psi_q(x)$ differs by $\pi$. More precisely, by inspecting equation~\eqref{PeriodicBloch_App} and using the expression~\eqref{PhiZak}, we see that the Zak phase computed with $\psi_q$ is 0 when $\Phi_n$ is $\pi$, and \emph{vice-versa}. Importantly, while the value of the Zak phase differ, the change of topological phase between two system is identical whether one uses $u_q$ of $\psi_q$.

	
	\section{Transfer matrix}\label{Sec-TM}
	
	Alternative proofs can be obtained through the usual transfer matrix. Here we list some useful results. Setting
	\begin{equation}
		\mathbf{U}(x,\omega)=\left(\begin{array}{c}
			u\\
			\ds E\,u^{'}
		\end{array}
		\right), \qquad
		\mathbf{A}(x,\omega)=\left(
		\begin{array}{cc}
			0 & 1/E(x) \\
			-\rho(x)\,\omega^2 & 0
		\end{array}
		\right),
	\end{equation}
	then the Helmholtz equation \eqref{Helmholtz} writes as a differential equation on $[0,h]$
	\begin{equation}
		\frac{d}{dx} \mathbf{U}=\mathbf{A}\,\mathbf{U},\qquad x \notin{\mathcal I},
		\label{ODE-Helmholtz}
	\end{equation} 
	with jump conditions \eqref{JCspringmass} at points of ${\cal I}$. Integration of \eqref{ODE-Helmholtz} on the subintervals $]x_j,x_{j+1}[$ and use of the jump conditions \eqref{JCspringmass} leads to
	\begin{equation}
		\mathbf{U}(h,\omega)=\mathbf{M}\,\mathbf{U}(0,\omega),
		\label{MatMh0}
	\end{equation}
	where $\mathbf{M}$ is the transfer matrix in ${\cal M}_2(\mathbb{R})$. Reciprocity and conservation of energy yields the general form
	\begin{equation}
		\mathbf{M}(\omega)=\left(
		\begin{array}{cc}
			\ds \alpha(\omega) & \beta(\omega) \\ [5pt]
			\ds \tilde{\beta}(\omega) & \alpha(\omega)
		\end{array}
		\right),\hspace{1cm} \det(\mathbf{M})=1.
		\label{PropMatM}
	\end{equation} 
	Bloch-Floquet theorem implies that $\mathbf{U}(h,\omega)=\lambda\,\mathbf{U}(0,\omega)$ with $\lambda=e^{\mathrm{i}q}$, where $q\in[0,\pm \pi]$ is the Bloch wavenumber. Comparison with \eqref{MatMh0} gives that $\lambda\equiv\lambda_{0,1}$ are the eigenvalues of $\mathbf{M}$. One also obtains the relations
	\begin{equation}
		\lambda_0+\lambda_1=2\alpha \in\mathbb{R},\hspace{1cm} \lambda_0\, \lambda_1=\alpha^2-\beta\tilde{\beta}=1.
		\label{LambdaEig}
	\end{equation}
	Depending on $\omega$, two cases occur. In the first case $|\lambda_{0,1}|=1$, the eigenvalues are complex conjugate, hence $q_{0,1}$ is purely real: it corresponds to a band. In the second case $|\lambda_{0,1}|\neq 1$, the eigenvalues are real and have the same sign, hence $q_{0,1}$ is purely imaginary: it corresponds to a gap. Some useful properties on the coefficients of $\mathbf{M}$ are now stated.
	
	\begin{Property}[Coefficient $\alpha$]
		The following properties hold:
		\begin{itemize}
			\item in bands: $|\alpha(\omega)|<1$;
			\item in gaps: $|\alpha(\omega)|>1$;
			\item on the edges of gaps: $|\alpha(\omega)|=1$.
		\end{itemize}
		\label{PropAlpha}
	\end{Property}
	
	\begin{proof}
		In a band, $\lambda_{0,1}$ are complex conjugates. From \eqref{LambdaEig}, it follows that $\alpha=\Re\mbox{e}(\lambda_0)\in]0,1[$, and hence $|\alpha|<1$. In a gap at $q=0$, then $\lambda_0=y\in]0,1[$ and $\lambda_1=1/y$. It follows $\alpha(y)=(y+1/y)/2$ and $\alpha^{'}(y)=(1-1/y^2)/2<0$. Since $\alpha(1)=1$, then $\alpha(y)>1$ for all $y<1$. The same argument holds in a gap at $q=\pi$. Lastly, at the band edges one has $\lambda_0=\lambda_1=\pm 1$, which concludes the proof.
	\end{proof}
	
	\begin{Property}[Coefficients $\beta$ and $\tilde{\beta}$]
		The following properties hold:
		\begin{itemize}
			\item in bands, $\beta(\omega)$ and $\tilde{\beta}(\omega)$ have opposite signs;
			\item in gaps, $\beta(\omega)$ and $\tilde{\beta}(\omega)$ have the same sign;
			\item $\beta(\omega)$ or $\tilde{\beta}(\omega)$ vanishes at edges of gaps.
		\end{itemize}
		\label{PropBeta}
	\end{Property}
	
	\begin{proof}
		From \eqref{LambdaEig}, it follows $\beta\,\tilde{\beta}=\alpha^2-1$. The different cases in Property \ref{PropAlpha} allow to conclude.
	\end{proof}
	
	\noindent
	From \eqref{PropMatM}, the components of the eigenvector $\mathbf{U}(0)$ of $\mathbf{M}$ satisfy
	$$
	\left\{
	\begin{array}{l}
		\beta\,E\,u^{'}(0,\omega)=(\lambda-\alpha)\,u(0,\omega),\\ [6pt]
		\tilde{\beta}u(0,\omega)=(\lambda-\alpha)\,E\,u^{'}(0,\omega),
	\end{array}
	\right.
	$$
	and hence
	$$
	Z(\omega)=\frac{\beta(\omega)}{\lambda(\omega)-\alpha(\omega)}.
	$$
	In gaps, Properties \ref{PropAlpha} and \ref{PropBeta} give $\alpha(\omega) \neq \lambda(\omega)$ and $\beta(\omega)\neq 0$. It implies that the surface impedance is real, finite, and never vanishes.

	
	%
	%


	\end{document}